\documentclass[onefignum,onetabnum]{siamonline190516}

\usepackage{bm}
\usepackage{amsfonts}
\usepackage{graphicx}

\usepackage{xcolor}

\ifpdf
\DeclareGraphicsExtensions{.eps,.pdf,.png,.jpg}
\else
\DeclareGraphicsExtensions{.eps}
\fi

\usepackage{enumitem}
\setlist[enumerate]{leftmargin=.5in}
\setlist[itemize]{leftmargin=.5in}

\usepackage[square,numbers,sort&compress]{natbib}

\newsiamremark{remark}{Remark}
\newsiamremark{hypothesis}{Hypothesis}
\crefname{hypothesis}{Hypothesis}{Hypotheses}
\newsiamthm{claim}{Claim}

\usepackage{amsopn}
\DeclareMathOperator{\diag}{diag}

\ifpdf
\hypersetup{
  pdftitle={The multiplex decomposition and applications},
  pdfauthor={R. Berner, V. Mehrmann, E. Schöll and S. Yanchuk}
}
\fi

\title{The multiplex decomposition: An analytic framework for multilayer dynamical networks. \thanks{Submitted to the editors \today.
		\funding{The authors acknowledge financial support by the Deutsche Forschungsgemeinschaft (DFG) (German Research Foundation)--Project Nos. 411803875 and 440145547.
}}}
\author{
	Rico Berner \footnotemark[3] \thanks{Institute of Theoretical Physics, Technische Universit\"at Berlin, Hardenbergstr. 36, 10623 Berlin, Germany, (\email{rico.berner@physik.tu-berlin.de}).} 
	\and Volker Mehrmann \thanks{Institute of Mathematics, Technische Universit\"at Berlin, Strasse des 17. Juni 136, 10623 Berlin, Germany.}
	\and Eckehard Schöll \footnotemark[2] \thanks{Bernstein Center for Computational Neuroscience, Berlin, Germany\and Potsdam Institute for Climate Impact Research, Germany.}
	\and Serhiy Yanchuk \footnotemark[3]
}

\headers{The multiplex decomposition}{Rico Berner, Volker Mehrmann, Eckehard Schöll, Serhiy Yanchuk}

\begin{document}

\maketitle

\begin{abstract}
Multiplex networks are networks composed of multiple layers such that the number of nodes in all layers is the same and the adjacency matrices between the layers are diagonal. We consider the special class of multiplex networks where the adjacency matrices for each layer are simultaneously triagonalizable. For such networks, we derive the relation between the spectrum of the multiplex network and the eigenvalues of the individual layers. As an application, we propose a generalized master stability approach that allows for a simplified, low-dimensional description of the stability of synchronized solutions in multiplex networks. We illustrate our result with a duplex network of FitzHugh-Nagumo oscillators. In particular, we show how interlayer interaction can lead to stabilization or destabilization of the synchronous state. Finally, we give explicit conditions for the stability of synchronous solutions in duplex networks of linear diffusive systems.
\end{abstract}



\section{Introduction}\label{intro}

Complex networks are well-established models in science and technology, with a wide range of applications from physics, chemistry, biology, neuroscience, as well as engineering and socio-economic systems~\cite{NEW03}. Much work has been devoted to understanding the statistical and topological properties of complex connectivity structures~\cite{ALB02a,COS07} and the collective dynamics on these structure~\cite{BOC18}. This study focuses on synchronization which is a particularly important type of collective dynamics~\cite{PIK01} playing an important role in the theory of complex dynamical systems, e.g. power grids or neural systems~\cite{BOC18,MOT13a}. One of the most powerful methodologies to study the synchronization on complex network structures is the master stability approach~\cite{PEC98}. Since its introduction, this methodology has been further developed and extended~\cite{BOC06a,SUN09a,DAH12,LAD13,NIS06a} and is still under continuous investigation~\cite{BER20b,BOE20,HAR19c,MUL20,NIC18}.

A recent focus in the field of complex dynamical networks are multilayer networks~\cite{BEL19,BOC14,DOM13,KIV14,POR20}. A prominent example are social networks, which can be described as groups of people with different patterns of contacts or interactions between them~\cite{AMA17a,GIR02}. Other applications include communication, supply and transportation networks, e.g., subway or airline networks~\cite{CAR13d}. In neuroscience, multilayer networks represent, for example, interactions of neurons by the transport of nutrition and metabolic resources~\cite{NIC17,VIR16c} or the modular connectivity structure of the human brain~\cite{BAT17,ROE19a,VAI18}. A special case of multilayer networks are multiplex topologies, where each layer contains the same set of nodes and there are only pairwise connections between corresponding nodes from different layers. These structures possess remarkable and analytically accessible properties that have been widely studied and used to understand real-world networks~\cite{BAZ20,BAT14,SOL13,BIA13,DOM16,MUC10,RAD13b}.

Collective dynamics on multilayer and multiplex networks has been extensively analyzed over the last years. Various forms of synchronization patterns have been described such as complete~\cite{GEN16,TAN19,SHA19}, cluster~\cite{ROS20}, intralayer~\cite{RAK20}, interlayer~\cite{LEY17a} and relay synchronization~\cite{DRA20,LEY18,SAW18,SAW18c,WIN19,SHA20b}. It was shown how the interaction of two layers can be used to induce and control certain types of dynamical phenomena such as phase cluster states~\cite{BER20}, solitary~\cite{MIK18} and chimera states~\cite{FRO18,MAK16,NIK19,OME19,RUZ20,RYB19}, explosive synchronization~\cite{ZHA15a}, congestion~\cite{SOL16}, stochastic and coherence resonance~\cite{SEM18,YAM20}. Recently, intimate relations between adaptive and multilayered coupling structures have been elucidated~\cite{BER19,BER19a,BER20c,KAS17,KAS18,PIT18}.

Other important topics are diffusive dynamics~\cite{CEN19a,GOM13,REQ16}, spreading processes~\cite{ARR18,SOR18,VAL15a}, and social interactions~\cite{GRA13} on multiplex networks. Here, in particular, the network structure characterized by the spectral properties is of crucial relevance~\cite{ARR17a,DOM16a,RAD13b,SOL13a}. For example, the maximal Laplacian eigenvalue can be directly linked to the epidemic threshold in models of infection spreading~\cite{SOR18}. By using perturbation methods, analytically interesting features for weak and strong interacting layers were found in~\cite{SOL13a}. However, little is known for the full range of interactions.

In this paper we introduce a method, which we call \emph{multiplex decomposition}. This framework allows for a rigorous description of the spectral properties of multiplex networks. For this, we establish a connection between the eigenvalues of multiplex networks and the eigenvalues of the individual layers for the special case when the adjacency matrices of individual layers are simultaneously triagonalizable. This holds in particular, if they are pairwise commuting. The multiplex decomposition allows us to greatly simplify the study of synchronization on multiplex networks. Namely, we show how the master stability function, which describes the stability of synchronization in the multiplex network, can be derived from the master stability function for the individual layer.

This article is organized as follows. In Section~\ref{sec:Sync} we provide a brief introduction to the synchronization problem in complex networks of coupled oscillators. Subsequently, in Section~\ref{sec:Mltplx_NW}, we introduce the  multiplex networks. In Section~\ref{sec:Mltplx_decomp} we present and analyze the multiplex decomposition method. In the subsequent Section~\ref{sec:Mltplx_decomp_appl} various applications of the multiplex decomposition are provided. More precisely, we establish a generalized framework for the master stability approach on multiplex networks, illustrate the generalized master stability approach for a network of FitzHugh-Nagumo oscillators, and provide an analytic description for the dynamics of a linear diffusive system on a duplex network. The results are summarized in Section~\ref{sec:conclusion}.


\section{Single layer networks: the synchronization problem}\label{sec:Sync}
We consider dynamical networks where a single layer is described by the following system
\begin{align}\label{eq:DynSys_OneLayer}
	\frac{\mathrm{d}\bm{x}_i}{\mathrm{d}t} = F(\bm{x}_i) - \sum_{j=1}^N a_{ij} H(\bm{x}_i-\bm{x}_j).
\end{align}
Here $\bm{x}_i=[x_{i,1},\dots,x_{i,d}]^T\in\mathbb{R}^d$ is the state vector of the $i$-th node, $i=1,\dots,N$, the network connectivity structure is given by the real $N\times N$ adjacency matrix $A=[a_{ij}]$. The functions $F$ and $H$ describe the local dynamics and the coupling between different nodes, respectively. We further assume $H(0)=0$. The coupling via the difference of the state variables $\bm{x}_i-\bm{x}_j$ is called diffusive and has been extensively studied in the literature, see, e.g., \cite{BEL05c,PEC98,POI19,YAN03b}.

For the network structure given by the adjacency matrix $A$, we only assume that it is (strongly) connected~\cite{KOR18}. However, for all examples given below, we use the particular class of nonlocally coupled ring networks given by $A=[a_{ij}]$ with
\begin{align}
	a_{ij}=\begin{cases}
		\kappa \quad\mbox{for}\ \ 0<| i-j| \le P \mbox{ or } | i-j| \ge N - P,\\
		0\quad\mbox{otherwise.}
	\end{cases}\label{eq:ring}
\end{align}
This means that any two oscillators on the ring are coupled if their indices $i$ and $j$ are separated at most by the coupling radius $P$. The adjacency matrix~(\ref{eq:ring}) defines a nonlocal ring structure with coupling range $P$ including two special cases: a local ring for $P=1$ and a globally coupled network for $P=N/2$ if $N$ is even, or $P=(N+1)/2$ otherwise. Note that self-couplings are excluded, since $a_{ii}=0$.
The adjacency matrix defined by (\ref{eq:ring}) is  a symmetric circulant matrix~\cite{DAV79} and therefore has constant row sum $\sum_{j=1}^{N}a_{ij}=2P\kappa$ for all $i=1,\dots,N$.

In the following, for numerical illustrations, we use networks of coupled FitzHugh-Nagumo oscillators that are well-known paradigmatic models for excitable neuronal systems~\cite{CHO18,GER20,RAM19,TAN18,YAN11}. Note that while the FitzHugh-Nagumo model was originally developed as a simplified model of a single neuron, it is also often used as a generic model for excitable media on a coarse-grained level~\cite{CHE07a,STE08}.

The local dynamics is given by
\begin{align}\label{eq:FHN_local}
	F([u,v]^T)
	= \begin{bmatrix}
		\frac{1}{\epsilon}\left(u - \frac{u^3}{3} - v\right)\\
		u + a
	\end{bmatrix},
\end{align}
and the coupling function between the nodes by
\begin{align}\label{eq:FHN_coupling}
	H([u,v]^T)
	=
		\begin{bmatrix}
	\frac{1}{\epsilon} \cos\varphi & \frac{1}{\epsilon} \sin\varphi \\
	-\sin\varphi & \cos\varphi
	\end{bmatrix}
	\begin{bmatrix}
		u\\v
	\end{bmatrix},
\end{align}
where $\epsilon$ describes the timescale separation between the fast activator variable $u$ and the slow inhibitor variable $v$~\cite{FIT61}. Depending on the threshold parameter $a$, each uncoupled node may exhibit excitable behavior ($\left| a \right| > 1$) or self-sustained limit cycle oscillations ($\left| a \right| < 1$), separated by a Hopf bifurcation at $\left| a \right| = 1$. We use the FitzHugh-Nagumo model in the oscillatory regime and fix the threshold parameter at $a=0.5$ sufficiently far from the Hopf bifurcation point.

The coupling function is chosen as a rotation matrix 
to parametrize the possibility of either diagonal coupling or activator-inhibitor cross-coupling by a single parameter $\varphi$.
For all simulations, we choose $\varphi = \frac{\pi}{2} - 0.1$, causing dominant activator-inhibitor cross-coupling, which is a commonly employed mechanism in biology~\cite{KIS04}.
Physically, this means that neuronal areas are coupled with a coupling phase lag $\varphi$~\cite{OME13}. The coupling phase has been shown to be crucial for the modeling of nontrivial partial synchronization patterns in the Kuramoto model~\cite{OME10a} and the FitzHugh-Nagumo model~\cite{OME13}.

Subsequently, we discuss complete synchronization defined as follows.
\begin{definition}
	A solution of \eqref{eq:DynSys_OneLayer} is called synchronous if $\bm{x}_i(t)=\bm{s}(t)$ for all $i=1,\dots,N$, all $t$ within a considered time interval (e.g. $\mathbb{R}$), and some function $\bm{s}(t)$.
\end{definition}
It is easily verified that system~\eqref{eq:DynSys_OneLayer} possesses a synchronous solution if $\bm{s}(t)$ solves $\dot{\bm{s}}=F(\bm{s})$. The local stability of the synchronous solution can be analyzed by studying the linearization of system~\eqref{eq:DynSys_OneLayer} around $\bm{s}(t)$. This linearized ($Nd$)-dimensional system is given by
%
\begin{align}\label{eq:CoupledSysDiff_OneLayer}
	\dot{\bm{\xi}}=\left[\mathbb{I}_{N}\otimes \mathrm{D}F(\bm{s})- L\otimes\mathrm{D}H(0)\right]\bm{\xi},
\end{align}
where $\otimes$ denotes the Kronecker product, see e.g. \cite{LIE15}, $\mathrm{D}$ denotes the derivative, and $\bm{\xi}=\bm{x}-\hat{1}_N\otimes\bm{s}$ with $\bm{x}=[\bm{x}_1,\dots,\bm{x}_N]^T, \hat{1}_N=[1,\dots,1]^T \in \mathbb R^N$. The \emph{Laplacian matrix} is defined as
\begin{align*}
	L=\begin{bmatrix}
		\sum_{j=1}^N a_{1j} & &\\
		& \ddots &\\
		& & \sum_{j=1}^N a_{Nj}
	\end{bmatrix}-A.
\end{align*}
In accordance with the master stability approach~\cite{PEC98}, system~\eqref{eq:CoupledSysDiff_OneLayer} can be block diagonalized if the Laplacian matrix $L$ is diagonalizable, i.e., there exists a unitary matrix $Q$ such that $L=Q D_L Q^H$, with diagonal matrix $D_L$. Here the superscript $H$ indicates the Hermitian conjugate. Hence, the local stability of the synchronous solution is determined by $N$ equations of dimension $d$
\begin{align}\label{eq:MSF_OneLayer}
	\dot{{\bm{\zeta}}}_i=\left[\mathrm{D}F(\bm{s})- \lambda_i \mathrm{D}H(0)\right]{\bm{\zeta}}_i,
\end{align}
where $\lambda_i \in \sigma(L)$, $i=1,\dots,N$   and ${\bm{\zeta}}=Q{\bm{\xi}}=[{\bm{\zeta}}_1,\dots,{\bm{\zeta}}_N]^T$. Here, by $\sigma(L)$ we denote the \emph{spectrum}, i.e. the set of eigenvalues of $L$, also called \emph{Laplacian eigenvalues} of the network. The local stability of the synchronous solution is determined by the largest Lyapunov exponent $\Lambda\in\mathbb{R}$ of system~\eqref{eq:MSF_OneLayer}. The solution is linearly stable if $\Lambda(\lambda_i)<0$ for all $i=1,\dots,N$; it is unstable if there exists at least one $\lambda_i$ such that $\Lambda(\lambda_i)>0$.
The \emph{master stability function} \cite{PEC98} is defined as
the largest Lyapunov exponent $\Lambda(\lambda)$, $\lambda\in \mathbb{C}$, of the system
$\dot{{{\zeta}}}=\left[\mathrm{D}F(\bm{s})-\lambda \mathrm{D}H(0)\right]{{\zeta}}$.
Once the master stability function is known, the stability of the synchronous solution can be deduced for any coupling structure by simply evaluating the master stability function at the points $\lambda_i$, $i=1,\dots,N$.

\begin{figure}\label{fig:FHN_Sync}
	\centering
	\includegraphics{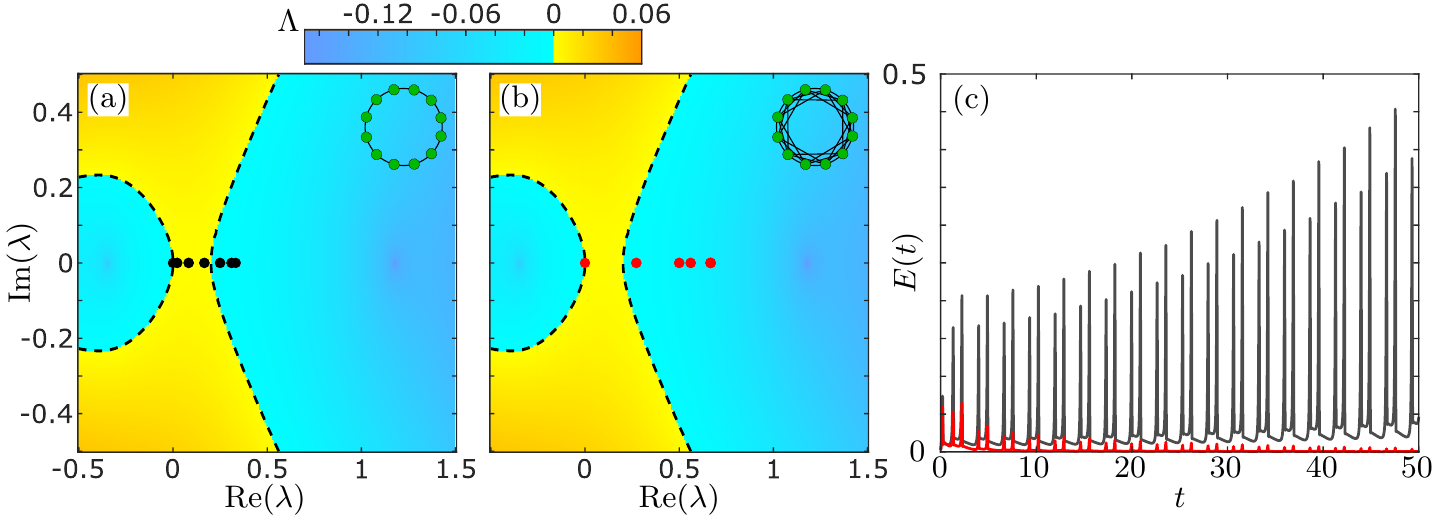}
	\caption{The synchronization problem in networks of FitzHugh-Nagumo oscillators. Panel (a,b) show the master stability function $\Lambda(\lambda)$ for a network of FitzHugh-Nagumo oscillators characterized by \eqref{eq:FHN_local} and \eqref{eq:FHN_coupling}. Additionally, in panel (a) and (b), the Laplacian eigenvalues are plotted as dots in black and red for the locally ($P=1$) and nonlocally ($P=3$) coupled ring networks illustrated in the inset in the upper right corner, respectively. Panel (c) shows the time series of the synchronization error $E(t)$ for simulations of coupled FitzHugh-Nagumo oscillators starting from a perturbed synchronized solution. The colors of the lines correspond to the colors of the dots in panel (a) and (b), respectively. Other parameters: $N=12$, $a=0.5$, $\epsilon=0.05$, $\varphi=\pi/2-0.1$, and $\kappa=1/N$.}
\end{figure}
In Figures~\ref{fig:FHN_Sync}(a,b), we present the numerically computed master stability function of a network of FitzHugh-Nagumo oscillators along with the Laplacian eigenvalues for two different network realizations. From the figures, we are able to infer that the synchronous solution is unstable for the locally coupled ring network in Figure~\ref{fig:FHN_Sync}(a) and stable for the nonlocally coupled ring network ($P=3$) in Figure~\ref{fig:FHN_Sync}(b). In order to verify the stability features, we numerically integrate the network equation
given by \eqref{eq:DynSys_OneLayer}, \eqref{eq:FHN_local}, and \eqref{eq:FHN_coupling}. For the simulation, we use a slightly perturbed synchronous solution as initial condition. For the visualization of the simulation results, we use the synchronization error $E(t)$ in order to quantify whether a network of oscillators achieves synchronization or not. The synchronization error is given by
\begin{align}\label{eq:SyncErr}
	E(t) = \sqrt{\sum_{m=1}^{d}\sum_{j=1}^N\left(x_{j,m}(t)-\frac{1}{N}\sum_{i=1}^{N}x_{i,m}\right)^2}.
\end{align}
Note that the value of $E(t)$ tends to zero if the solution of Equation~\eqref{eq:DynSys_OneLayer} tends to the synchronous solution. In Figure~\ref{fig:FHN_Sync}(c) it is clearly visible that the numerical solution diverges from (black line) and converges to (red line) the synchronous solution for the locally coupled network and nonlocally coupled ring network, respectively. This is in accordance with the master stability function.

In this section, we have outlined the interplay of the network structure and the dynamics with regards to the synchronous collective behavior of interacting agents. In the following section, we lift this problem of synchronization to multiplex network structures and show how the results for single layers can be used to obtain synchronization conditions for multiplex networks.

\section{Multiplex networks}\label{sec:Mltplx_NW}
Multilayer networks are networks where the whole set of nodes is divided into subsets which are said to belong together for various reasons. The induced subnetworks are then called layers. From the mathematical perspective, multilayer networks are simply networks. However, the special structure of a multilayer network, i.e., the partition into several subsets of nodes, has recently been considered to be very important in order to describe the dynamics on real-world networks~\cite{DE15,DIA16}. For  reviews on multilayer networks and their mathematical description, we refer to~\cite{BOC14,DOM13,KIV14}.

In the following, we consider so-called
\emph{multiplex networks} which form a particular class of multilayer networks. These consist of $M\in\mathbb{N}$ layers of $N$ nodes where the connectivity structure within the layer is given by adjacency matrices $A_\ell$ ($\ell=1,\dots,M$). These adjacency matrices determine the (directed) intralayer network structures. The connectivity between the layers is determined by the (directed) interlayer coupling structure expressed by an $M\times M$ matrix
\begin{align}\label{eq:interlayerCouplStruc}
	\mathcal{S} = \begin{bmatrix}
		0& \kappa_{1,2} & \cdots & \kappa_{1,M} \\
		\kappa_{2,1} & \ddots &\ddots & \vdots \\
		\vdots & \ddots & \ddots & \kappa_{M-1,M} \\
		\kappa_{M,1} & \cdots & \kappa_{M,M-1} & 0
	\end{bmatrix}.
\end{align}
If two layers $k$ and $\ell$ are connected, the corresponding entry $\kappa_{k,\ell}\ne 0$ indicates the strength of their connection. We restrict ourselves to multiplex networks, where all layers are connected in a one-to-one manner which is expressed by the identity matrix $\mathbb{I}_N$. An illustration of a duplex ($M=2$) and a triplex ($M=3$) network is provided in Figure~\ref{fig:MltPlx_Illu}. Without loss of generality, we assume that the layers do not have self-interconnection, i.e., $\kappa_{\ell,\ell}=0$ for all $\ell$.
\begin{figure}\label{fig:MltPlx_Illu}
	\centering
	\includegraphics{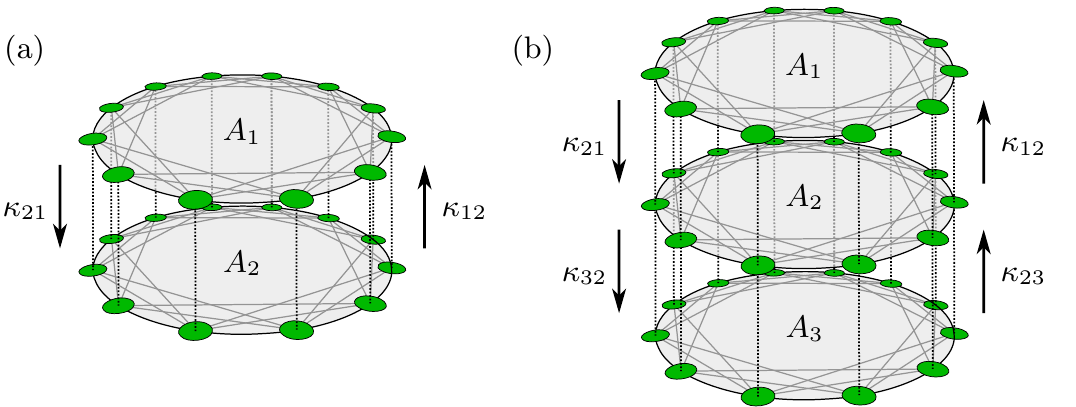}
	\caption{Illustrations of a duplex ($M=2$) (a) and a triplex ($M=3$) (b) network. In example (b), the layers $A_1$ and $A_3$ are not connected, i.e., $\kappa_{13}=\kappa_{31}=0$.}
\end{figure}
A general algebraic representation of multilayer networks can be achieved by multilinear forms, i.e., tensor structures~\cite{DOM13}. By flattening these tensors, however, a representation via an adjacency matrix is obtained. Here flattening means that one can relate a finite dimensional tensor space to another finite dimensional vector space via an isomorphism, see~\cite{KIV14}. In this context, the flattened representation 
takes the following block matrix form
\begin{align}\label{eq:multiplex_adjMtx}
	A^{(M)} = \begin{bmatrix}
		A_1& \kappa_{1,2}\mathbb{I} & \cdots & \kappa_{1,M}\mathbb{I} \\
		\kappa_{2,1}\mathbb{I} & \ddots &\ddots & \vdots \\
		\vdots & \ddots & \ddots & \kappa_{M-1,M}\mathbb{I} \\
		\kappa_{M,1}\mathbb{I} & \cdots & \kappa_{M,M-1} \mathbb{I} & A_M
	\end{bmatrix}=\mathcal{A}+\mathcal{S}\otimes \mathbb{I},
\end{align}
where $\mathcal{A}=\mathrm{diag}\left( A_1,\dots,A_M\right)$.

Although the described classes of multiplex networks possess a relatively simple form, they are far from completely understood from an analytic point of view. In the next section, we analyze the spectrum of such multiplex networks.

\section{The multiplex decomposition}\label{sec:Mltplx_decomp}
In this section, using the spectra $\sigma(A_k)$ of the  individual layers, we analyze the spectrum $\sigma(A^{(M)})$ for the class of multiplex networks introduced in the previous section.
We recall the basic fact that any square complex matrix $A$ is unitarily similar to an upper triangular matrix $T_A$, i.e. $A = Q T_A Q^H$, \emph{the Schur form}, see e.g.~\cite{LIE15}.
Note further that if  $A$ is \emph{normal}, i.e. $A^HA=AA^H$, then the Schur form is even diagonal. Further, we call two matrices $A$ and $B$ simultaneously triagonalizable if there exists a unitary matrix $Q$ such that $T_A = Q^H A Q$ and $T_B = Q^H B Q$ are triangular matrices. In particular, if $A$ and $B$ commute, they are simultaneously triagonalizable.
With these preliminaries, we can state the following result.
\begin{proposition}\label{prop:DeterminantDecompDiagAble}
If the set of matrices $A_{\ell,k}\in \mathbb{C}^{N\times N}$, $k,\ell=1,\dots,M$ are simultaneously triagonalizable, then the block matrix
	\begin{align}\label{eq:DeterminantDecompDiagAble}
		\mathbf{A} =
		 \begin{bmatrix}
			A_{1,1} & \cdots & A_{1,M}\\
			\vdots & \ddots & \vdots\\
			A_{M,1} & \cdots & A_{M,M}
		\end{bmatrix}
	\end{align}
	is unitarily similar to
	\begin{align*}
			\mathbf{T} =
	 \begin{bmatrix}
		T_{A_{1,1}} & \cdots & T_{A_{1,M}}\\
		\vdots & \ddots & \vdots\\
		T_{A_{M,1}} & \cdots & T_{A_{M,M}}
	\end{bmatrix},
\end{align*}
	where the blocks $T_{A_{k,\ell}}$ are the (common) Schur forms of the matrices $A_{k,\ell}$ with respect to a common unitary matrix $Q$. More specifically,
	\begin{align}
	\mathbf{A} = \mathrm{diag}\left(Q,\cdots,Q\right) \ \mathbf{T} \ \mathrm{diag}\left( Q^H,\cdots,Q^H\right).
	\end{align}
\end{proposition}
\begin{proof}
	All $A_{k,\ell}$ can be transformed to (a common) Schur form with the same $Q$ \cite{LIE15}, i.e. $A_{k,\ell} = Q T_{A_{k,\ell}} Q^H$.
The assertion then follows by applying the block diagonal matrices $\mathrm{diag}\left( Q,\cdots,Q\right)$ and $\mathrm{diag}\left( Q^H,\cdots,Q^H\right)$ from the left and right, respectively.
\end{proof}
For the multiplex network~\eqref{eq:multiplex_adjMtx}, Proposition~\ref{prop:DeterminantDecompDiagAble} yields the following result.
\begin{corollary}\label{cor:MultiplexDecomp}
	Let $A^{(M)}$ be the multiplex adjacency matrix~\eqref{eq:multiplex_adjMtx} with simultaneously triagonalizable matrices $A_{\ell}$, $\ell=1,\dots,M$, and the interlayer coupling matrix $\mathcal{S}$ as in \eqref{eq:interlayerCouplStruc}.
	Then, $A^{(M)}$ is unitarily similar to
		\begin{align}
			\label{BBB}
			\mathbf{B} =
\begin{bmatrix}
	B_{1,1} & B_{1,2} & \cdots & B_{1,N}\\
	0 & \ddots & \ddots & \vdots\\
	\vdots & \ddots & \ddots & B_{N-1,N}\\
	0 & \cdots & 0 & B_{N,N}
\end{bmatrix},
	\end{align}
	where
	\begin{align}
		\label{BBB111}
	B_{i,i} = \begin{bmatrix}
		(\lambda_{1})_i & & 0\\
		& \ddots & \\
		0 &  & (\lambda_{M})_i
	\end{bmatrix} + \mathcal{S}=
\begin{bmatrix}
	(\lambda_{1})_i & \kappa_{1,2} & \cdots & \kappa_{1,M}\\
	\kappa_{2,1} & \ddots & \ddots & \vdots\\
	\vdots & \ddots & \ddots & \kappa_{M-1,M}\\
	\kappa_{M,1} & \cdots & \kappa_{M,M-1} & (\lambda_{M})_i
\end{bmatrix}
,
\end{align}
and $(\lambda_{\ell})_i$ is the $i$th eigenvalue of $A_{\ell}$ corresponding to a (common) Schur form $A_\ell$ under a common unitary transformation $Q$.

In particular, $\sigma(A^{(M)})$ is given as the union of the spectra of the $N$ matrices $B_{i,i}$, $i=1,\dots,N$, which are of dimension $M\times M$, and are functions of the eigenvalues $(\lambda_l)_i$ of the layer Laplacians, i.e.
\begin{equation}
\label{eq:mltplxDecomp}
\sigma (A^{(M)}) = \bigcup_{i=1}^N
\sigma \left(
\begin{bmatrix}
(\lambda_{1})_i & & 0\\
& \ddots & \\
0 &  & (\lambda_{M})_i
\end{bmatrix} + \mathcal{S}
\right).
\end{equation}
\end{corollary}
\begin{proof}
Since all $A_{\ell}$ and $\kappa_{\ell,k}\mathbb{I}$ are simultaneously triagonalizable, Proposition~\ref{prop:DeterminantDecompDiagAble} implies that $A^{(M)}$ is unitarily similar to the matrix
	\begin{align}
		\label{DDD}
\begin{bmatrix}
		T_{A_{1}} & & 0\\
		& \ddots & \vdots\\
		0 & & T_{A_{M}}
	\end{bmatrix} + \mathcal{S} \otimes \mathbb{I}.
\end{align}
By employing the perfect-shuffle permutation matrix, see e.g. \cite{LIE15},
we see that \eqref{DDD} is unitarily similar to \eqref{BBB}--\eqref{BBB111}. The off-diagonal matrices $B_{i,j}$ ($i=1,\dots,N-1; N \ge j>i$) consist of the off-diagonal elements of the matrices $T_{A_k}$. The triangular block matrix form is therefore inherited from the matrices $T_{A_k}$.
\end{proof}

We call the relation~\eqref{eq:mltplxDecomp} the \emph{multiplex decomposition}. Note that the multiplex decomposition is in general not restricted to adjacency matrices alone. In fact, it can be used for Laplacian or more general types of matrices as long as the individual matrices $A_{\ell}$ are simultaneously triagonalizable.

In the special case that all layers are identical i.e. $L_1=\cdots=L_M=L$ with eigenvalues $\lambda_1,\dots,\lambda_N$ and the eigenvalues of $\mathcal{S}$ are given by $\rho_1,\dots,\rho_M$, then
\[
\sigma(A^{(M)}) = \{ \lambda_i+\rho_j \in \mathbb{C}:  \lambda_i\in \sigma(L),\rho_j\in \sigma(\mathcal{S})  \}.
\]
and  each eigenvalue of the multiplex network has the simple form
$\lambda_i + \rho_j$ for some $i=1,\dots,N$ and $j=1,\dots,M$.
This observation is a special case of the well-known Theorem of Stephanos, see e.g. \cite{LIE15}.

In the following, we illustrate the multiplex decomposition for a duplex system. Suppose that $A_1,A_2\in\mathbb{R}^{N\times N}$ and $\kappa_{ij}\in\mathbb{R}$ ($i,j=1,2$). Then, the $2N\times2N$ block matrix
\begin{align}\label{eq:DuplexAdj}
	A^{(2)} = \begin{bmatrix}
		A_1 & \kappa_{12}\mathbb{I}\\
		\kappa_{21}\mathbb{I} & A_2
	\end{bmatrix}
\end{align}
describes a duplex network and we have the following
direct consequence of Corollary~\ref{cor:MultiplexDecomp}.
\begin{corollary}\label{cor:EigenvaluePolynomsMultiplex}
	Suppose that two matrices $A_1,A_2\in\mathbb{R}^{N\times N}$ are simultaneously triagonalizable. Then,
	\begin{align}\label{eq:DuplexEigXXX}
		\sigma( A^{(2)}) =
		\bigcup_{i=1}^{N}
		\sigma \begin{bmatrix}
			(\lambda_1)_i  & \kappa_{12}\\
			\kappa_{21} & (\lambda_2)_i
		\end{bmatrix},
	\end{align}	
i.e.,	the eigenvalues of the duplex network $A^{(2)}$ can be found by solving the $N$ quadratic  equations
	\begin{align}\label{eq:DuplexDecomp}
		\mu^2-\left((\lambda_1)_i+(\lambda_2)_i\right)\mu+(\lambda_1)_i(\lambda_2)_i-\kappa_{12}\kappa_{21} = 0, \quad i=1,\dots,N,
	\end{align}
	where $(\lambda_{1})_i$ and $(\lambda_{2})_i$ are the corresponding eigenvalues of $A_1$ and $A_2$ that are ordered according to their (common) Schur forms resulting from a common unitary transformation $Q$.
\end{corollary}
If the duplex network corresponds to a master-slave configuration, i.e., either $\kappa_{12}=0$ or $\kappa_{21}=0$, then the eigenvalues of the duplex network are given by the eigenvalues of the individual layers $\sigma(A^{(2)})=\sigma(A_1)\cup \sigma(A_2)$. Analogous statements hold for master-slave configurations with arbitrary number of layers.

In summary, for the discussed special class of multiplex networks,  the spectrum is determined by the eigenvalues of the corresponding layers. In the following, we use these results to simplify the master stability approach for multiplex networks.

\section{Applications of the multiplex decomposition}\label{sec:Mltplx_decomp_appl}
In this section we provide two perspectives where the derived multiplex decomposition can be used to generalize existing results and make others analytically accessible.
\subsection{The master stability approach for multiplex networks: general results}\label{sec:MSA_MltPlx}
Since the introduction of the master stability approach~\cite{PEC98}, this methodology has been successfully used to describe the synchronization phenomena in complex networks~\cite{BOC06a,DAH12,LAD13} and is even today under constant investigation~\cite{BER20b,BOE20,MUL20}. In~\cite{SOL13a,TAN19}, the master stability function for dynamical systems on multiplex networks was analyzed for a diffusive system of the form
\begin{align}\label{eq:DynSys_MltPlx}
	\frac{\mathrm{d}\bm{x}_i^\ell}{\mathrm{d}t} = F(\bm{x}_i^\ell) -  \sum_{j=1}^N a^\ell_{ij} H(\bm{x}_i^\ell-\bm{x}_j^\ell) - \sum_{k=1}^M \kappa_{\ell,k} G(\bm{x}_i^\ell-\bm{x}_i^k)
\end{align}
where $\bm{x}_i^\ell\in\mathbb{R}^d$ is the state vector of the $i$th node of the $\ell$th layer whose connectivity structure is given by the entries  of the real (intralayer) adjacency matrix $A_\ell=[a^\ell_{ij}]$. The functions $F$, $H$ and $G$ describe the local dynamics, the coupling between the systems of the same layer, and the coupling between the  systems of different layers, respectively. We further assume that $G(0)=H(0)=0$. The values $\kappa_{k,\ell}$ ($k\ne \ell$) are interlayer coupling constants.

The synchronous solution 
of Equation~\eqref{eq:DynSys_MltPlx} is given by a solution $\bm{s}(t)$ of $\dot{\bm{s}}=F(\bm{s})$, and the master stability function can be obtained from the variational equation
\begin{align}\label{eq:CoupledSysDiff}
	\dot{\bm{\xi}}=\left(\mathbb{I}_{NM}\otimes DF(\bm{s})-\mathcal{L}^{\text{intra}}\otimes DH(0)- \mathcal{L}^{\text{inter}}\otimes DG(0)\right)\bm{\xi}.
\end{align}
Here, $\bm{\xi}=\bm{x}-\hat{1}_M\otimes\hat{1}_N\otimes\bm{s}$ and $\bm{x}=[\bm{x}^1,\dots,\bm{x}^M]^T$ with $\bm{x}^\ell=[\bm{x}_1^\ell,\dots,\bm{x}_N^\ell]^T$. The \emph{intralayer Laplacian} is defined as
\[
	\mathcal{L}^{\text{intra}}=\diag( L_{1}, \ldots, L_{M})
\]
with
\begin{align*}
	L_{\ell}=\begin{bmatrix}
		\sum_{j=1}^N a^{\ell}_{1,j} & &\\
		& \ddots &\\
		& & \sum_{j=1}^N a^{\ell}_{N,j}
	\end{bmatrix}-A_{\ell}.
\end{align*}
The \emph{ interlayer Laplacian} is defined as
\begin{align}
	\mathcal{L}^{\text{inter}}= {\mathcal L}_{\mathcal{S}}\otimes\mathbb{I}_N = \left(\begin{bmatrix}
		\sum_{\ell=1}^M \kappa_{1,\ell} & &\\
		& \ddots &\\
		& & \sum_{\ell=1}^M \kappa_{M,\ell}
	\end{bmatrix} - \mathcal{S}\right) \otimes \mathbb{I}_N,
\end{align}
where $\mathcal{S}$ is the \emph{ interlayer coupling structure} as in~\eqref{eq:multiplex_adjMtx}. Further details on the system~\eqref{eq:CoupledSysDiff} are given in \cite{TAN19}, where it was shown that if $\mathcal{L}^{\text{intra}}$ and $\mathcal{L}^{\text{inter}}$ commute, then the master stability equation for system~\eqref{eq:CoupledSysDiff} reads
\begin{align*}
	\dot{\bm{\zeta}_i}=\left[DF(\bm{s})-\lambda_i DH(0)-\gamma_j DG(0)\right]\bm{\zeta}_i,
\end{align*}
where $\bm{\zeta}_i\in\mathbb{C}^d$,  and $\lambda_i$ and $\gamma_j$ are  eigenvalues of $\mathcal{L}^{\text{intra}}$ and $\mathcal{L}^{\text{inter}}$, respectively. In the special case $DH(0)=DG(0)$, the master stability equation can be reduced to
\begin{align}\label{eq:MSEComposite}
	\dot{\bm{\zeta}}=\left[DF(\bm{s})-(\lambda_i+\gamma_j) DH(0)\right]\bm{\zeta}.
\end{align}
Equation~\eqref{eq:MSEComposite} is called the \emph{master stability equation for the composite system},  where a single supra-Laplacian matrix~$\mathcal{L}^{\text{supra}}=\mathcal{L}^{\text{intra}}+\mathcal{L}^{\text{inter}}$ describes the network topology~\cite{DOM13,KIV14}.

Note that a sufficient but not necessary condition that the matrices $\mathcal{L}^{\text{intra}}$ and $\mathcal{L}^{\text{inter}}$ commute is that the layers are identical $L_{1}=\cdots=L_{M}$.
Using Corollary~\ref{cor:MultiplexDecomp} for the composite system, the master stability function can be analyzed under  weaker  conditions.
\begin{proposition}\label{prop:MSE_general}
	Consider the variational equation~\eqref{eq:CoupledSysDiff} with coupling functions satisfying $DH(0)=DG(0)$. Suppose further, that all $L_{\ell}$, $\ell=1,\dots,M$ are simultaneously triagonalizable, with corresponding $i$th eigenvalues  $(\lambda_\ell)_i$ of the Laplacian matrix $L_\ell$ (in its common Schur form).
	Then, the local stability of \eqref{eq:CoupledSysDiff} is implied by the local stability of the $N$  systems
	\begin{align}\label{eq:MSEComposite_mild}
		\dot{\bm{\zeta}}=\left[DF(\bm{s})-\psi_i DH(0)\right]\bm{\zeta},
	\end{align}
	for $\psi_i=\psi((\lambda_1)_i,\dots,(\lambda_M)_i)$, $i=1,\dots,N$, where $(\lambda_\ell)_i$ is the $i$th eigenvalue of the Laplacian matrix $L_\ell$ and the  mappings $\psi_i=\psi((\lambda_{1})_i,\dots,(\lambda_{M}))_i$ are defined via the eigenvalues of the matrix 
	\begin{align}\label{eq:masterMapping}
		\left(\begin{bmatrix}
			(\lambda_{1})_i-\psi_i & & 0\\
			& \ddots & \\
			0 &  & (\lambda_{M})_i-\psi_i
		\end{bmatrix} + {\mathcal L}_\mathcal{S}\right),
	\end{align}
	i.e., for $i=1,\ldots,N$, $\psi_i\in \sigma \left(
	\mathrm{diag}\left(
	(\lambda_{1})_i,\dots, (\lambda_{M})_i
	\right)
	+ {\mathcal L}_\mathcal{S}
	\right).$
	%
\end{proposition}
\begin{proof}
	Using $DH(0)=DG(0)$, the variation equation for~\eqref{eq:CoupledSysDiff} on the synchronous solution $\bm{s}(t)$ is given by
	\begin{align*}
		\dot{\bm{\xi}}=\left[\mathbb{I}_{NL}\otimes DF(\bm{s})-\mathcal{L}^{\text{supra}}\otimes DH(0)\right]\bm{\xi},
	\end{align*}
	By assumption all $L_{\ell}$ ($\ell=1,\dots,M$) are pairwise commuting and $\mathcal{L}^{\text{inter}}$ is an $M\times M$ block matrix with blocks that are  $N\times N$ identity matrices multiplied by scalars. Hence, Proposition~\ref{prop:DeterminantDecompDiagAble} can be applied to $\left(\mathcal{L}^{\text{supra}}-\mathbb{I}_{M\cdot N}\right)$. Similarly to Corollary~\ref{cor:MultiplexDecomp}, the eigenvalues are those of
	$\sigma \left[
	\mathrm{diag}\left(
	(\lambda_{1})_i,\dots, (\lambda_{M})_i
	\right)
	+ { \mathcal L}_\mathcal{S}
	\right].$
	In contrast to Corollary~\ref{cor:MultiplexDecomp}, here the matrix
	${\mathcal L}_\mathcal{S}$ appears instead of $\mathcal{S}$, since, due to the diffusive coupling, the interlayer coupling $\mathcal{L}^{\text{inter}}$ involves the Laplacian ${\mathcal L}_\mathcal{S}$. The eigenvalues are then determined from the characteristic equation
	$\det\left(\mathcal{L}^{\text{supra}}-\lambda \mathbb{I}_{M\cdot N}\right)$ as in~\eqref{eq:masterMapping}.
\end{proof}
Proposition \ref{prop:MSE_general} yields a powerful tool to investigate not only the influence of the multiplex network structure on the stability of the synchronous solution but also the impact of different layer topologies.

As an example, we consider duplex systems with simultaneously triagonalizable $L_1$ and $L_2$. Then, the supra-Laplacian has the form
\begin{align*}
	\mathcal{L}^{\text{supra}}=\begin{bmatrix}
		 L_{1}+\kappa_{12}\mathbb{I} & -\kappa_{12}\mathbb{I}\\
		- \kappa_{21}\mathbb{I} &  L_{2}+\kappa_{21}\mathbb{I}
	\end{bmatrix}.
\end{align*}
Knowing the eigenvalues of $L_{1}$ and $L_{2}$, the master stability function parameter 
is determined using Equation~\eqref{eq:masterMapping}.
Equation~\eqref{eq:masterMapping} possesses four solutions.

Since $L_{1}$ and $L_{2}$ are Laplacian matrices, they each have at least one zero eigenvalue corresponding to the $N$-dimensional eigenvector $\hat{1}=[1,\dots,1]^T$. 
Denote the other eigenvalues by $\lambda_1$ and $\lambda_2$, respectively.

As a result, we obtain two solutions of \eqref{eq:masterMapping}: $\mu_{1}=0$ and $\mu_{2}=\kappa_{1,2}+\kappa_{2,1}$. The first value $\mu_1=0$ corresponds to the (multiplex) neutral eigenvector $[\hat{1},\hat{1}]^T$ and the second is induced by the duplex structure and is completely independent of the individual layer topologies. It corresponds to the eigenvector $[\hat{1},-\hat{1}]^T$.
The other eigenvalues $\mu_{3,4}$ of the supra-Laplacian matrix $\mathcal{L}^{\text{supra}}$ are given by the  nonlinear mappings
\begin{align}\label{eq:masterParamDuplex}
	\mu_{3,4}(\lambda_{1},\lambda_{2}) = \frac{{\lambda}_{1}+{\lambda}_{2}+\kappa_{12}+\kappa_{21}}{2}
	\pm\frac12\sqrt{
		\begin{aligned}
			\left({\lambda}_{1}-{\lambda}_{2}+\kappa_{12}-\kappa_{21}\right)^2+ 4\kappa_{12}\kappa_{21}
		\end{aligned}.
	}
\end{align}
%

In the following, we consider two special cases.
First, we assume that there is no connection from the second to the first layer. Then we have a master-slave set-up which means that $\kappa_{12}=0$. With this, the master stability function parameter is $\mu_3(\lambda_{1},\lambda_{2}) =\lambda_{1}$ and $\mu_4(\lambda_{1},\lambda_{2}) = \lambda_{2}+ \kappa_{21}$.
 In this case, the stability of a synchronous solution in the duplex network is reduced to the pure one-layer system. The stability in the duplex system is determined by the spectrum of the individual layer topologies where only in the second layer the spectrum is shifted by $\kappa_{21}$ due to the interaction.

The second case starts from the consideration in~\cite{TAN19} and assumes 
identical layers $L_1=L_2$, and hence $\lambda=(\lambda_1)_i=(\lambda_2)_i$ for all $i=1,\dots,N$.
Taking this into account, the equations for the master stability function parameter~\eqref{eq:masterParamDuplex} yield
\begin{align*}
	\mu_\ell(\lambda_{1},\lambda_{2}) &=\lambda+\rho_{\ell},\quad \ell=3,4,
\end{align*}
where $\rho_{3}=0$ and $\rho_4=\kappa_{12}+\kappa_{21}$ are the eigenvalues of ${\mathcal L}_\mathcal{S}$. Hence, $\rho_{1,2}$ are also eigenvalues of the supra-Laplacian,  since $L_1$ possess at least one zero eigenvalue. Note that due to the zero eigenvalue of the Laplaciancs $L_i$, it holds in general that the eigenvalues of ${\mathcal L}_\mathcal{S}$ are part of the eigenspectrum of $\mathcal{L}^\text{supra}$, even if the layer Laplacians are not simultaneously triagonalizable~\cite{SOL13a}. In this sense, we have generalized this statement and shown that the eigenvalues of multiplex networks with identical layers are exactly the eigenvalues of the single layers shifted by the eigenvalues of the interlayer Laplacian.
\subsection{FitzHugh-Nagumo oscillators on duplex networks}\label{sec:MSA_FHN}
In this section, we show how the simplified master stability approach of Section~\ref{sec:MSA_MltPlx} can be applied to the duplex network~\eqref{eq:DuplexAdj} of coupled FitzHugh-Nagumo oscillators as introduced in Section~\ref{sec:Sync}.

\begin{figure}\label{fig:FHN_SingleToDuplex}
	\centering
	\includegraphics{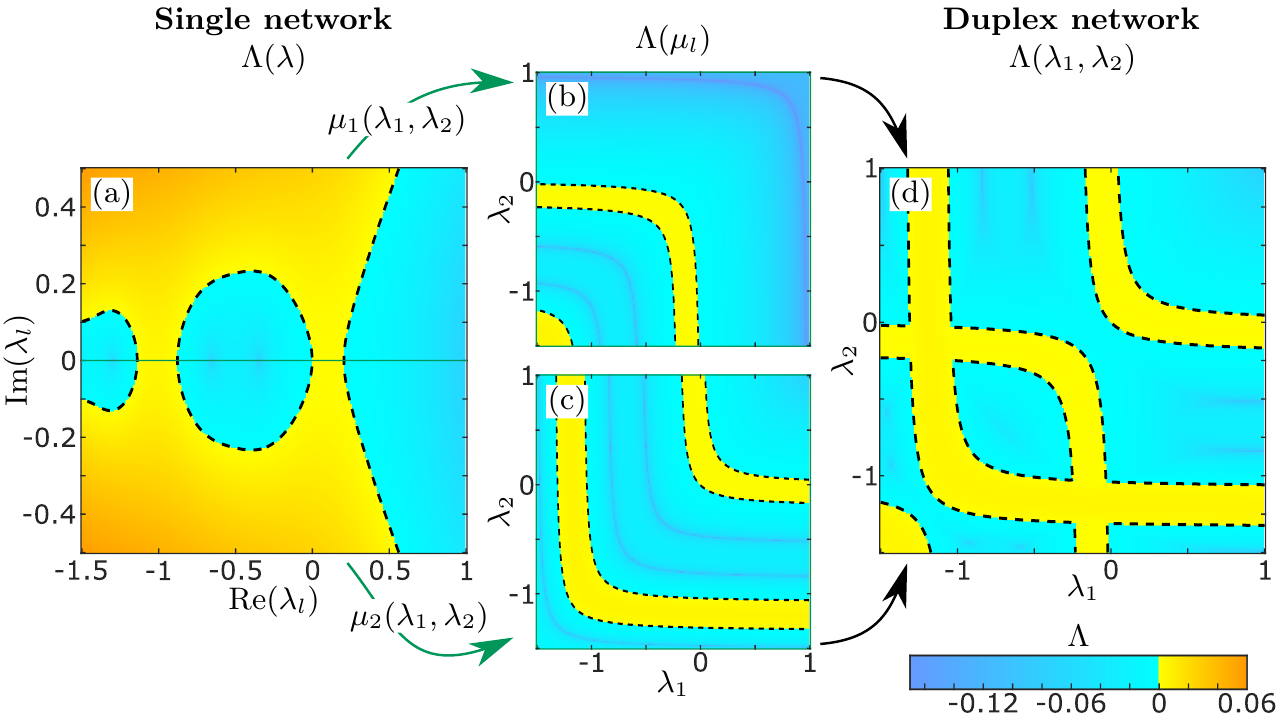}
	\caption{Derivation of the master stability function for the duplex network with undirected layer networks. Panel (a) shows the master stability function $\Lambda(\lambda)$ for a network of FitzHugh-Nagumo oscillators  \eqref{eq:FHN_local}--\eqref{eq:FHN_coupling}. In panel (b,c) two (partial) master stability functions $\Lambda(\mu_j(\lambda_1,\lambda_2))$, $j=3,4$ are shown, where the mappings $\mu_j(\lambda_1,\lambda_2)$ are given by \eqref{eq:masterParamDuplex} for real $\lambda_1$ and $\lambda_2$. Panel (d) shows the simplified master stability function of the duplex network depending on real eigenvalues $\lambda_j$ of the Laplacian for the individual layers. Parameters: $a=0.5$, $\epsilon=0.05$, $\varphi=\pi/2-0.1$, $\kappa_{12}=\kappa_{21}=0.2$}
\end{figure}
 When considering the layers ($l=1,2$) of the duplex separately, the master stability function $\Lambda(\lambda_l)$ can be calculated. This function takes the same form for each layer and, for our illustrative example of coupled FitzHugh-Nagumo oscillators, it is shown in Figure~\ref{fig:FHN_SingleToDuplex}(a). We recall that the master stability function does not depend on the network structure.

Consider further two layers whose networks are undirected, i.e., the layers possess symmetric adjacency matrices. Then, the Laplacian eigenvalues $\lambda_l$ for each layer are real. Restricting ourselves to undirected networks for each layer, we can use the two mappings $\mu_3(\lambda_1,\lambda_2)$ and $\mu_4(\lambda_1,\lambda_2)$ as given in Equation~\eqref{eq:masterParamDuplex} to calculate the master stability parameters for the duplex network.

 Using \eqref{eq:masterParamDuplex} for real $\lambda_{1,2}$ and the master stability function in Figure~\ref{fig:FHN_SingleToDuplex}(a), we obtain the two (partial) master stability functions $\Lambda(\mu_3(\lambda_1,\lambda_2))$ and $\Lambda(\mu_4(\lambda_1,\lambda_2))$ for duplex networks depending on the eigenvalues in the individual layers. These two (partial) master stability functions are shown in Figures~\ref{fig:FHN_SingleToDuplex}(b,c).
 
 Finally,  the master stability function of the duplex network is given
 as the maximum
 $\Lambda_{\mathrm{duplex}}(\lambda_1,\lambda_2)=\max_{j=3,4} \Lambda(\mu_j(\lambda_1,\lambda_2))$,
see Figure~\ref{fig:FHN_SingleToDuplex}(d).
Importantly, the obtained duplex master stability function depends only on the eigenvalues of the individual layers.

\begin{figure}\label{fig:FHN_SimplifiedDuplexExamples}
	\centering
	\includegraphics{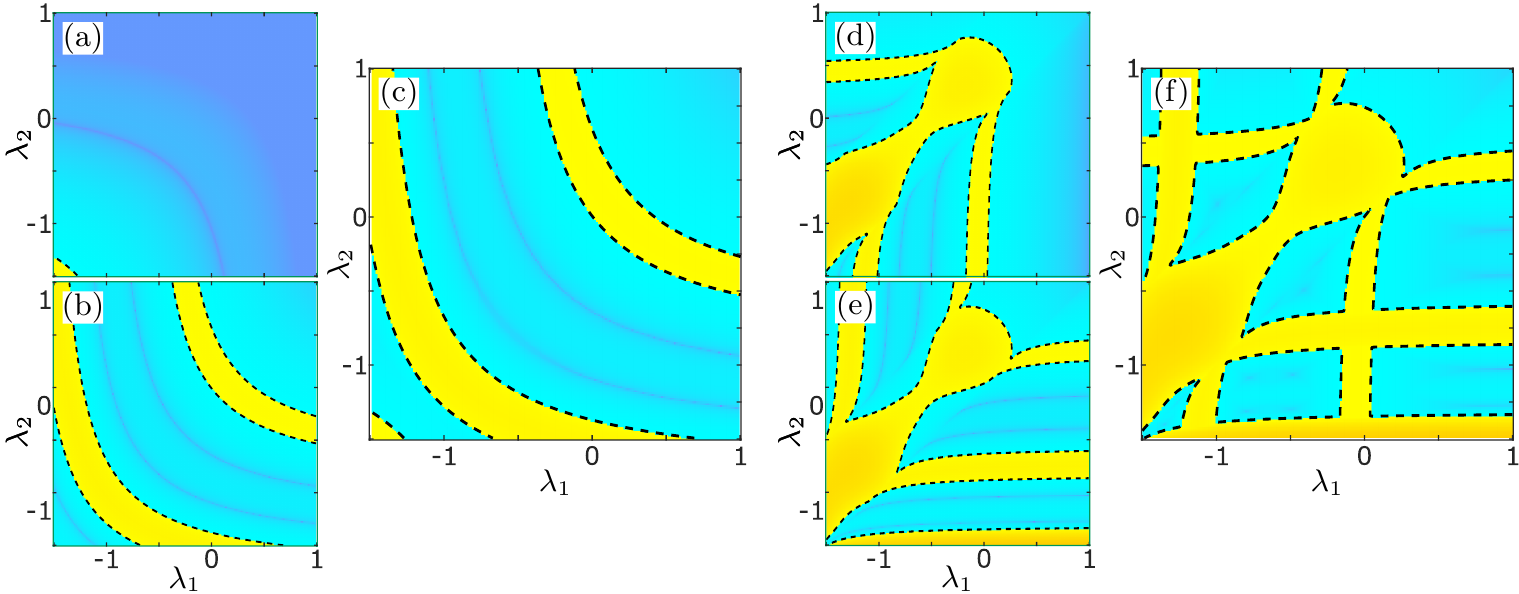}
	\caption{The partial and the simplified master stability function of the duplex network with undirected layer networks for different interlayer coupling strengths. Panels (a,b) and (d,e) show the two (partial) master stability functions corresponding to panels (c) and (f), respectively. In panels (c) and (f), the simplified master stability functions of the duplex network are shown for $\kappa_{12}=0.7, \kappa_{21}=0.9$ and $\kappa_{12}=0.2, \kappa_{21}=-0.3$, respectively. The color code is the same as in \ref{fig:FHN_SingleToDuplex}. Other parameters: $a=0.5$, $\epsilon=0.05$, $\varphi=\pi/2-0.1$.}
\end{figure}

The discussed approach simplifies the application of the master stability function for duplex networks significantly. Instead of computing $2N$ eigenvalues for the duplex network for each combination of any two layers individually, we provide a complete mapping that depends on the characteristic of the individual layers. Therefore, the multiplex decomposition allows for a reduction from a multilayered system to a system with only a single layer. The master stability approach, moreover, reduces the single layer dynamical system even further to the dynamics of a single node. We note that the simplified master stability function for the duplex network shown in Figure~\ref{fig:FHN_SingleToDuplex}(d) depends on the interlayer coupling strengths $\kappa_{12}$ and $\kappa_{21}$. In Figure~\ref{fig:FHN_SimplifiedDuplexExamples}, we provide two different examples of the master stability function of the duplex network for two particular choices of the interlayer coupling strengths.
\begin{figure}\label{fig:FHN_StabDestab}
	\centering
	\includegraphics{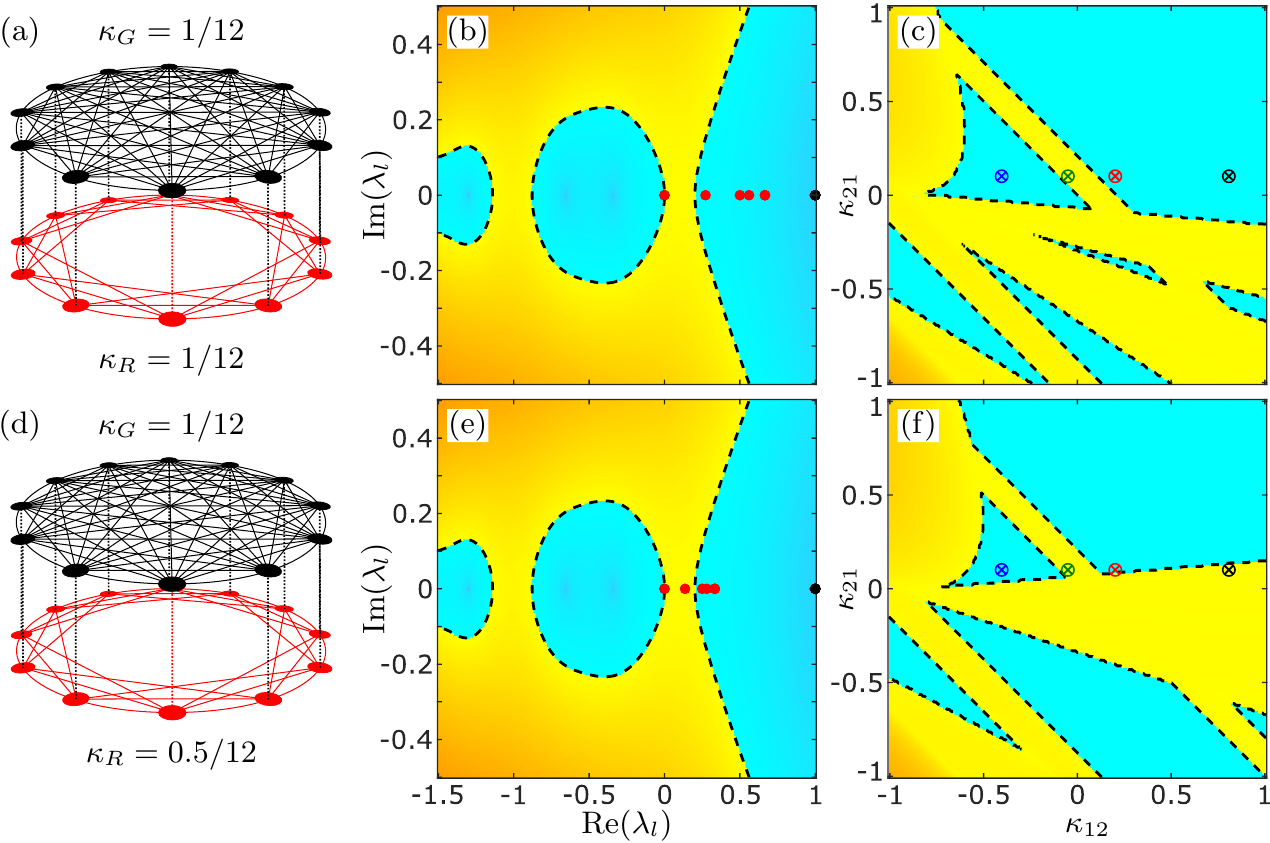}
	\caption{Stability of the synchronous solution depending on the interlayer coupling structure. Panels (a) and (d) show the duplex network structure considered for the panels (b,c) and (e,f), respectively. Panels (b) and (e) show the master stability function for a network of FitzHugh-Nagumo oscillators characterized by \eqref{eq:FHN_local} and \eqref{eq:FHN_coupling} along with the Laplacian eigenvalues of networks shown in (a) and (d). Panels (c) and (f) show the maximal Lyapunov exponents for the synchronous solution depending on the interlayer coupling strengths $\kappa_{12}$ and $\kappa_{21}$. The crosses corresponds to the parameter values used for the analysis in Figure~\ref{fig:FHN_Sync_StabDestab}. The color code is the same as in Fig.~\ref{fig:FHN_SingleToDuplex}. Other parameters: $a=0.5$, $\epsilon=0.05$, and $\varphi=\pi/2-0.1$.}
\end{figure}
We have derived the simplified master stability function for a duplex network for fixed values of the interlayer coupling strengths $\kappa_{12}$ and $\kappa_{21}$. Hence, for two given simultaneously triagonalizable layers and undirected network structure, we are able to determine the local stability of the synchronous solution by looking at the Lyapunov exponents in Figure~\ref{fig:FHN_SingleToDuplex}(c) for the $2N$ Laplacian eigenvalues of the individual layers.

In Figure~\ref{fig:FHN_StabDestab} we show how the stability can change
with the variation of the coupling strengths.
We consider two duplex networks consisting of two layers. The first layer is a globally coupled network where the coupling strengths  for all links is equal to $\kappa_G$. The second layer is a nonlocally coupled ring network with $P=3$, with all nontrivial coupling weights equal to $\kappa_R$. The destabilization with varying $\kappa_G$ and $\kappa_R$ can be simply analyzed using the introduced duplex master stability function.

The upper panel with Figures~\ref{fig:FHN_StabDestab}(a,b,c) corresponds to the
intralayer coupling \eqref{eq:ring} with  $\kappa_G=\kappa_R=1/12$ for both layers. The lower panel corresponds to the coupling strengths \eqref{eq:ring} with
$\kappa_G=1/12$ for the first layer and $\kappa_R=1/24$ for the second one.

Figures~\ref{fig:FHN_StabDestab}(b,e) show the master stability function $\Lambda(\alpha+\mathrm{i}\beta)$ for a network of FitzHugh-Nagumo oscillators. The function is the same for both cases.
Additionally, the Laplacian eigenvalues $\lambda_\ell$ of the individual layers are plotted as colored nodes.
 Note that the eigenvalues are scaled by the intralayer coupling strengths $\kappa_G$ and $\kappa_R$. From Figure~\ref{fig:FHN_StabDestab}(b) it is clear that the single layers corresponding to Figure~\ref{fig:FHN_StabDestab}(a) possess stable synchronous solutions if they are uncoupled, i.e., $\kappa_{12}=\kappa_{21}=0$. For the lower panel, Figure~\ref{fig:FHN_StabDestab}(e) shows that the synchronous solution is linearly unstable for the uncoupled layers.

Figures~\ref{fig:FHN_StabDestab}(c,f) show the maximal Lyapunov exponents for the respective networks in Figures~\ref{fig:FHN_StabDestab}(a,d) depending on the interlayer coupling strengths $\kappa_{12}$ and $\kappa_{21}$. Even though the individual layers of Figures~\ref{fig:FHN_StabDestab}(a) would possess stable synchronous solutions for $\kappa_{12}=\kappa_{21}=0$, there are interlayer coupling structures that yield a positive Lyapunov exponent in Figures~\ref{fig:FHN_StabDestab}(c). Thus, the interaction of two layers can give rise to destabilization of the synchronous solution. In turn the opposite effect can be also obtained. While the second layer in Figures~\ref{fig:FHN_StabDestab}(d) alone would possess an unstable synchronous solution, there are regions in Figures~\ref{fig:FHN_StabDestab}(f) with negative maximal Lyapunov exponent. Hence, the interaction of two layers may stabilize the synchronous solution.
\begin{figure}\label{fig:FHN_Sync_StabDestab}
	\centering
	\includegraphics{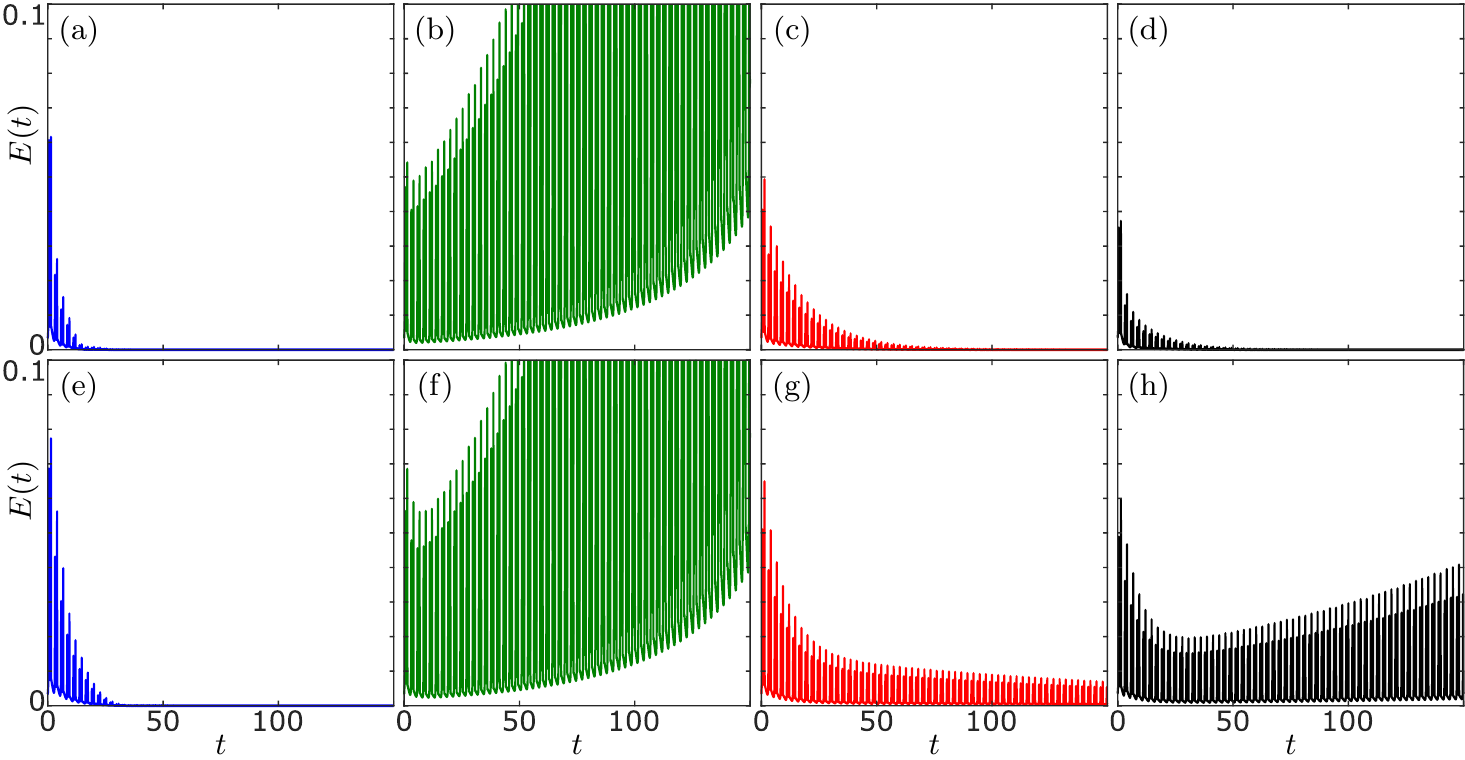}
	\caption{Time series of the synchronization error $E(t)$ for simulations of coupled FitzHugh-Nagumo oscillators starting from a perturbed synchronized solution. The colors of the lines correspond to the colors of the dots in Figure~\ref{fig:FHN_StabDestab}(c,f). Panels (a-d) correspond to Figure~\ref{fig:FHN_StabDestab}(f) and panels (e-h) correspond to Figure~\ref{fig:FHN_StabDestab}(f). All parameters are as in Figure~\ref{fig:FHN_StabDestab}.}
\end{figure}
In order to verify the findings in Figures~\ref{fig:FHN_StabDestab}(c,f), we integrate the coupled system of FitzHugh-Nagumo oscillators numerically for different values of the in\-ter\-layer coup\-ling strength, where we start from a slightly perturbed synchronous solution. The considered  values for $\kappa_{12}$ and $\kappa_{21}$ are displayed in Figures~\ref{fig:FHN_StabDestab}(c,f). As in Fig~\ref{fig:FHN_Sync}, the synchronization error~\eqref{eq:SyncErr} is used to show the synchronization of the network where $E=0$ corresponds to the synchronous solution. The results of the simulation are presented in Figure~\ref{fig:FHN_Sync_StabDestab}. In Figure~\ref{fig:FHN_Sync_StabDestab}, the first row corresponds to simulation of the network shown in Figure~\ref{fig:FHN_StabDestab}(a) and the second row corresponds to simulation of the network shown in  Figures~\ref{fig:FHN_StabDestab}(d). The numerical results verify our analytical findings and the effects of the destabilization and stabilization of synchronous solutions through multiplexing.

Summarizing, in this section we have employed the novel multiplex master stability function, developed in Sec.~\ref{sec:MSA_MltPlx}, to a duplex network of FitzHugh-Nagumo oscillators. We have shown explicitly how a master stability function for a duplex network can be derived. Further, we have investigated the effect of the duplex structure on the stability of the synchronous solution, thus extending the findings of~\cite{BER20} to a more complex system and we have shown that multiplexing can destabilize and stabilize synchronous solutions. In the following section, we show how the multiplex decomposition can be used to understand linear diffusive systems.

\subsection{Analytic treatment of linear diffusive dynamics on multiplex networks}\label{sec:DiffSys}
In the previous section we have considered the dynamics of linear systems given by variational equation~\eqref{eq:CoupledSysDiff}. In the context of diffusive systems on complex networks, linear diffusive processes have been considered recently to study the dynamics on social and transport networks~\cite{BAR11d,GOM13}. In~\cite{GOM13}, the authors investigated a duplex system ($L=2$) with $F=0$, $H=G=\mathbb{I}_d$, and $\kappa=\kappa_{12}=\kappa_{21}$.

Let us assume that the supra-Laplacian matrix is given by
\begin{align*}
	\mathcal{L}^{\text{supra}}=\begin{bmatrix}
		 L_{1}+\kappa\mathbb{I} & -\kappa\mathbb{I}\\
		-\kappa \mathbb{I} &  L_{2}+ \kappa\mathbb{I}
	\end{bmatrix},
\end{align*}
where $L_{1}$ and $L_{2}$ are simultaneously triagonalizable. Due to the structure of the supra-Laplacian, we can apply~Proposition~\ref{prop:MSE_general}. In accordance with~\cite{GOM13} and our former findings in Section~\ref{sec:MSA_MltPlx}, we have the two eigenvalues $\mu_1=0$ and $\mu_2=2\kappa$ corresponding to the neutral eigenvector $(\hat{1},\hat{1})^T$ and the eigenvector $(\hat{1},-\hat{1})^T$, respectively. The other eigenvalues are given as solution to the equations~\eqref{eq:DuplexDecomp} and read
\begin{align*}
	(\mu_i)_{3,4} = \frac{(\lambda_1)_i+(\lambda_2)_i+2\kappa}{2}
	\pm\frac12\sqrt{
		\begin{aligned}
			\left((\lambda_1)_i-(\lambda_2)_i\right)^2+ 4\kappa^2
		\end{aligned}
	},
\end{align*}
where $(\lambda_1)_i$ and $(\lambda_2)_i$ are the nonzero eigenvalues of the Laplacians $L_{1}$ and $L_{2}$, respectively.
Recall that the order of the eigenvalues is given by the specific transformation $Q$ in Proposition~\ref{prop:MSE_general}. With this, under the assumption that $L_1$ and $L_2$ are simultaneously triagonalizable, we have derived an analytic expression for the spectra of the diffusive system  in~\cite{GOM13}. In contrast to~\cite{GOM13}, here, we do not need any perturbation methods and make the dynamics analytically accessible for the full range of interlayer coupling $\kappa$.

\section{Conclusion}\label{sec:conclusion}
In this article, we have developed a novel analysis of the spectral structure of multiplex networks. The new approach has a broad range of applications to physical, biological, socio-economic, and technological systems, ranging from plasticity in neurodynamics~\cite{BER20} or the dynamics of linear diffusive systems~\cite{GOM13,SOL13a} to generalizations of the master stability function~\cite{PEC98,TAN19} for adaptive networks~\cite{BER20b}.

Our multiplex decomposition allows for the understanding of multiplex networks by the features of their individual layers alone. In Sec.~\ref{sec:MSA_MltPlx}, we have derived the general framework for the master stability approach of multiplex networks. We have shown how the Laplacian eigenvalues of the individual layers determine the local stability of the synchronous solution for the multiplex networks if all layers are simultaneously triagonalizable. A special case is given by pairwise commuting layers. Subsequently, in Sec.~\ref{sec:MSA_FHN}, the framework has been applied to duplex networks of coupled FitzHugh-Nagumo oscillators. For undirected duplex networks, we have derived the explicit form of the simplified master stability function. Moreover, we have analyzed the stability of the synchronous solution with varying interlayer coupling structure. We have shown analytically and numerically that multiplexing may stabilize or destabilize synchronous solutions.

As another application of the multiplex decomposition, in Sec.~\ref{sec:DiffSys} we have discussed the dynamics of linear diffusive system on a duplex network. Complementing the analysis of~\cite{GOM13,SOL13a}, we have shown how the eigenvalues of the supra-Laplacian depend on the full range of the interlayer coupling strength. 
By this, we go beyond the perturbation approach used in~\cite{GOM13,SOL13a}.

We note that the restriction of simultaneously triagonalizable layer matrices can be slightly lifted by using existing perturbative methods~\cite{RES06b,GOM13}. Further, we point out that the developed simplification of the master stability approach for multiplex network is not restricted to the specific form of dynamical systems that has been used in this article. In particular, our developed approach can also be used for multiplex systems with delay~\cite{DAH12}, an adaptive network structure~\cite{BER20b}, non-smooth coupling functions~\cite{NIC18}, and  for nearly identical dynamical systems~\cite{SUN09a}.

\section*{Acknowledgments}
We thank Jakub Sawicki for stimulating discussions.


\begin{thebibliography}{10}
	
	\bibitem{ALB02a}
	{\sc R.~Albert and A.~L. Barab{\'a}si}, {\em Statistical mechanics of complex
		networks}, Rev. Mod. Phys., 74 (2002), pp.~47--97,
	\url{https://doi.org/10.1103/revmodphys.74.47}.
	
	\bibitem{AMA17a}
	{\sc R.~Amato, N.~E. Kouvaris, M.~San~Miguel, and A.~D{\'i}az-Guilera}, {\em
		Opinion competition dynamics on multiplex networks}, New J. Phys., 19 (2017),
	p.~123019.
	
	\bibitem{BAR11d}
	{\sc M.~Barth{\'e}lemy}, {\em Spatial {N}etworks}, Phys. Rep., 499 (2011),
	pp.~1--101, \url{https://doi.org/10.1016/j.physrep.2010.11.002}.
	
	\bibitem{BAT17}
	{\sc F.~Battiston, V.~Nicosia, M.~Chavez, and V.~Latora}, {\em Multilayer motif
		analysis of brain networks}, Chaos, 27 (2017), p.~047404.
	
	\bibitem{BAT14}
	{\sc F.~Battiston, V.~Nicosia, and V.~Latora}, {\em Structural measures for
		multiplex networks}, Phys. Rev. E, 89 (2014), p.~032804,
	\url{https://doi.org/10.1103/physreve.89.032804}.
	
	\bibitem{BAZ20}
	{\sc M.~Bazzi, L.~G.~S. Jeub, A.~Arenas, S.~D. Howison, and M.~A. Porter}, {\em
		A framework for the construction of generative models for mesoscale structure
		in multilayer networks}, Phys. Rev. Research, 2 (2020), p.~023100,
	\url{https://doi.org/10.1103/physrevresearch.2.023100}.
	
	\bibitem{BEL19}
	{\sc I.~V. Belykh, D.~Carter, and R.~Jeter}, {\em Synchronization in multilayer
		networks: When good links go bad}, SIAM J. Appl. Dyn. Syst., 18 (2019),
	pp.~2267--2302, \url{https://doi.org/10.1137/19m1257123}.
	
	\bibitem{BEL05c}
	{\sc I.~V. Belykh, M.~Hasler, M.~Lauret, and H.~Nijmeijer}, {\em
		Synchronization and graph topology}, Int. J. Bifurc. Chaos, 15 (2005),
	p.~3423.
	
	\bibitem{BER19a}
	{\sc R.~Berner, J.~Fialkowski, D.~V. Kasatkin, V.~I. Nekorkin, S.~Yanchuk, and
		E.~Sch{\"o}ll}, {\em Hierarchical frequency clusters in adaptive networks of
		phase oscillators}, Chaos, 29 (2019), p.~103134,
	\url{https://doi.org/10.1063/1.5097835}.
	
	\bibitem{BER20c}
	{\sc R.~Berner, A.~Polanska, E.~Sch{\"o}ll, and S.~Yanchuk}, {\em Solitary
		states in adaptive nonlocal oscillator networks}, Eur. Phys. J. Spec. Top.,
	229 (2020), pp.~2183--2203,
	\url{https://doi.org/https://doi.org/10.1140/epjst/e2020-900253-0}.
	
	\bibitem{BER20}
	{\sc R.~Berner, J.~Sawicki, and E.~Sch{\"o}ll}, {\em Birth and stabilization of
		phase clusters by multiplexing of adaptive networks}, Phys. Rev. Lett., 124
	(2020), p.~088301, \url{https://doi.org/10.1103/physrevlett.124.088301}.
	
	\bibitem{BER19}
	{\sc R.~Berner, E.~Sch{\"o}ll, and S.~Yanchuk}, {\em Multiclusters in networks
		of adaptively coupled phase oscillators}, SIAM J. Appl. Dyn. Syst., 18
	(2019), pp.~2227--2266, \url{https://doi.org/10.1137/18m1210150}.
	
	\bibitem{BER20b}
	{\sc R.~Berner, S.~Vock, E.~Sch{\"o}ll, and S.~Yanchuk}, {\em Desynchronization
		transitions in adaptive networks}, Phys. Rev. Lett., 126 (2021), p.~028301,
	\url{https://doi.org/10.1103/physrevlett.126.028301}.
	
	\bibitem{BIA13}
	{\sc G.~Bianconi}, {\em Statistical mechanics of multiplex networks: Entropy
		and overlap}, Phys. Rev. E, 87 (2013), p.~062806,
	\url{https://doi.org/10.1103/physreve.87.062806}.
	
	\bibitem{BOC14}
	{\sc S.~Boccaletti, G.~Bianconi, R.~Criado, C.~I. del Genio,
		J.~G\'omez-Garde\~nes, M.~Romance, I.~Sendi\~{n}a Nadal, Z.~Wang, and
		M.~Zanin}, {\em The structure and dynamics of multilayer networks}, Phys.
	Rep., 544 (2014), pp.~1--122,
	\url{https://doi.org/10.1016/j.physrep.2014.07.001}.
	
	\bibitem{BOC06a}
	{\sc S.~Boccaletti, V.~Latora, Y.~Moreno, M.~Chavez, and D.~U. Hwang}, {\em
		Complex networks: {S}tructure and dynamics}, Phys. Rep., 424 (2006),
	pp.~175--308, \url{https://doi.org/doi: 10.1016/j.physrep.2005.10.009}.
	
	\bibitem{BOC18}
	{\sc S.~Boccaletti, A.~N. Pisarchik, C.~I. del Genio, and A.~Amann}, {\em
		Synchronization: {F}rom Coupled Systems to Complex Networks}, Cambridge
	University Press, Cambridge, 2018.
	
	\bibitem{BOE20}
	{\sc R.~B{\"o}rner, P.~Schultz, B.~{\"U}nzelmann, D.~Wang, F.~Hellmann, and
		J.~Kurths}, {\em Delay master stability of inertial oscillator networks},
	Phys. Rev. Research, 2 (2020), p.~023409,
	\url{https://doi.org/10.1103/physrevresearch.2.023409}.
	
	\bibitem{CAR13d}
	{\sc A.~Cardillo, M.~Zanin, J.~G{\`o}mez Garde\~nes, M.~Romance, A.~Garcia~del
		Amo, and S.~Boccaletti}, {\em Modeling the multi-layer nature of the european
		air transport network: {R}esilience and passengers re-scheduling under random
		failures}, Eur. Phys. J. ST, 215 (2013), pp.~23--33.
	
	\bibitem{CEN19a}
	{\sc G.~Cencetti and F.~Battiston}, {\em Diffusive behavior of multiplex
		networks}, New J. Phys., 21 (2019), p.~035006,
	\url{https://doi.org/10.1088/1367-2630/ab060c}.
	
	\bibitem{CHE07a}
	{\sc A.~Chernihovskyi and K.~Lehnertz}, {\em Measuring synchronization with
		nonlinear excitable media}, Int. J. Bifurc. Chaos, 17 (2007), pp.~3425--3429.
	
	\bibitem{CHO18}
	{\sc T.~Chouzouris, I.~Omelchenko, A.~Zakharova, J.~Hlinka, P.~Jiruska, and
		E.~Sch{\"o}ll}, {\em Chimera states in brain networks: empirical neural vs.
		modular fractal connectivity}, Chaos, 28 (2018), p.~045112,
	\url{https://doi.org/https://doi.org/10.1063/1.5009812}.
	
	\bibitem{COS07}
	{\sc L.~d.~F. Costa, F.~A. Rodrigues, G.~Travieso, and P.~R. Villas~Boas}, {\em
		Characterization of complex networks: A survey of measurements}, Adv. Phys.,
	56 (2007), pp.~167--242, \url{https://doi.org/10.1080/00018730601170527}.
	
	\bibitem{DAH12}
	{\sc T.~Dahms, J.~Lehnert, and E.~Sch{\"o}ll}, {\em Cluster and group
		synchronization in delay-coupled networks}, Phys. Rev. E, 86 (2012),
	p.~016202, \url{https://doi.org/10.1103/physreve.86.016202}.
	
	\bibitem{DAV79}
	{\sc P.~J. Davis}, {\em Circulant matrices}, Wiley, 1979.
	
	\bibitem{ARR17a}
	{\sc G.~F. de~Arruda, E.~Cozzo, T.~P. Peixoto, F.~A. Rodrigues, and Y.~Moreno},
	{\em Disease localization in multilayer networks}, Phys. Rev. X, 7 (2017),
	p.~011014, \url{https://doi.org/10.1103/physrevx.7.011014}.
	
	\bibitem{ARR18}
	{\sc G.~F. de~Arruda, F.~A. Rodrigues, and Y.~Moreno}, {\em Fundamentals of
		spreading processes in single and multilayer complex networks}, Physics
	Reports, 756 (2018), pp.~1--59,
	\url{https://doi.org/10.1016/j.physrep.2018.06.007}.
	
	\bibitem{DOM16a}
	{\sc M.~De~Domenico, C.~Granell, M.~A. Porter, and A.~Arenas}, {\em The physics
		of spreading processes in multilayer networks}, Nat. Phys., 12 (2016),
	pp.~901--906, \url{https://doi.org/10.1038/nphys3865}.
	
	\bibitem{DE15}
	{\sc M.~De~Domenico, V.~Nicosia, A.~Arenas, and V.~Latora}, {\em Structural
		reducibility of multilayer networks}, Nat. Commun., 6 (2015), p.~6864,
	\url{https://doi.org/10.1038/ncomms7864}.
	
	\bibitem{DOM16}
	{\sc M.~De~Domenico, S.~Sasai, and A.~Arenas}, {\em Mapping multiplex hubs in
		human functional brain networks}, Front. Neurosci., 10 (2016), p.~326,
	\url{https://doi.org/10.3389/fnins.2016.00326}.
	
	\bibitem{DOM13}
	{\sc M.~De~Domenico, A.~Sol{\'e}-Ribalta, E.~Cozzo, M.~Kivel{\"a}, Y.~Moreno,
		M.~A. Porter, S.~G\'omez, and A.~Arenas}, {\em Mathematical formulation of
		multilayer networks}, Phys. Rev. X, 3 (2013), p.~041022.
	
	\bibitem{GEN16}
	{\sc C.~I. del Genio, J.~G\'omez-Garde\~nes, I.~Bonamassa, and S.~Boccaletti},
	{\em Synchronization in networks with multiple interaction layers}, Sci.
	Adv., 2 (2016), p.~e1601679, \url{https://doi.org/10.1126/sciadv.1601679}.
	
	\bibitem{ROS20}
	{\sc F.~Della~Rossa, L.~M. Pecora, K.~Blaha, A.~Shirin, I.~Klickstein, and
		F.~Sorrentino}, {\em Symmetries and cluster synchronization in multilayer
		networks}, Nat. Commun., 11 (2020), p.~3179,
	\url{https://doi.org/10.1038/s41467-020-16343-0}.
	
	\bibitem{DIA16}
	{\sc M.~Diakonova, V.~Nicosia, V.~Latora, and M.~San~Miguel}, {\em
		Irreducibility of multilayer network dynamics: the case of the voter model},
	New J. Phys., 18 (2016), p.~023010,
	\url{https://doi.org/10.1088/1367-2630/18/2/023010}.
	
	\bibitem{DRA20}
	{\sc F.~Drauschke, J.~Sawicki, R.~Berner, I.~Omelchenko, and E.~Sch{\"o}ll},
	{\em Effect of topology upon relay synchronization in triplex neuronal
		networks}, Chaos, 30 (2020), p.~051104,
	\url{https://doi.org/https://doi.org/10.1063/5.0008341}.
	
	\bibitem{FIT61}
	{\sc R.~FitzHugh}, {\em Impulses and physiological states in theoretical models
		of nerve membrane}, Biophys. J., 1 (1961), pp.~445--466.
	
	\bibitem{FRO18}
	{\sc N.~S. Frolov, V.~A. Maksimenko, V.~V. Makarov, D.~Kirsanov, A.~E. Hramov,
		and J.~Kurths}, {\em Macroscopic chimeralike behavior in a multiplex
		network}, Phys. Rev. E, 98 (2018), p.~022320,
	\url{https://doi.org/10.1103/physreve.98.022320}.
	
	\bibitem{GER20}
	{\sc M.~Gerster, R.~Berner, J.~Sawicki, A.~Zakharova, A.~Skoch, J.~Hlinka,
		K.~Lehnertz, and E.~Sch{\"o}ll}, {\em {FitzHugh-Nagumo} oscillators on
		complex networks mimic epileptic-seizure-related synchronization phenomena},
	Chaos, 30 (2020), p.~123130, \url{https://doi.org/10.1063/5.0021420}.
	
	\bibitem{GIR02}
	{\sc M.~Girvan and M.~E.~J. Newman}, {\em Community structure in social and
		biological networks}, Proc. Natl. Acad. Sci. USA, 99 (2002), p.~7821.
	
	\bibitem{GOM13}
	{\sc S.~G\'omez, A.~D{\'i}az-Guilera, J.~G\'omez-Garde\~nes, C.~J.
		P\'{e}rez~Vicente, Y.~Moreno, and A.~Arenas}, {\em Diffusion dynamics on
		multiplex networks}, Phys. Rev. Lett., 110 (2013), p.~028701,
	\url{https://doi.org/10.1103/physrevlett.110.028701}.
	
	\bibitem{GRA13}
	{\sc C.~Granell, S.~G\'omez, and A.~Arenas}, {\em Dynamical interplay between
		awareness and epidemic spreading in multiplex networks}, Phys. Rev. Lett.,
	111 (2013), p.~128701, \url{https://doi.org/10.1103/physrevlett.111.128701}.
	
	\bibitem{HAR19c}
	{\sc J.~D. Hart, Y.~Zhang, R.~Roy, and A.~E. Motter}, {\em Topological control
		of synchronization patterns: Trading symmetry for stability}, Phys. Rev.
	Lett., 122 (2019), p.~058301,
	\url{https://doi.org/10.1103/physrevlett.122.058301}.
	
	\bibitem{KAS18}
	{\sc D.~V. Kasatkin and V.~I. Nekorkin}, {\em Synchronization of chimera states
		in a multiplex system of phase oscillators with adaptive couplings}, Chaos,
	28 (2018), p.~093115, \url{https://doi.org/10.1063/1.5031681}.
	
	\bibitem{KAS17}
	{\sc D.~V. Kasatkin, S.~Yanchuk, E.~Sch{\"o}ll, and V.~I. Nekorkin}, {\em
		{S}elf-organized emergence of multi-layer structure and chimera states in
		dynamical networks with adaptive couplings}, Phys. Rev. E, 96 (2017),
	p.~062211, \url{https://doi.org/10.1103/physreve.96.062211}.
	
	\bibitem{KIS04}
	{\sc M.~A. Kiskowski, M.~S. Alber, G.~L. Thomas, J.~A. Glazier, N.~B.
		Bronstein, J.~Pu, and S.~A. Newman}, {\em Interplay between
		activator-inhibitor coupling and cell-matrix adhesion in a cellular automaton
		model for chondrogenic patterning}, Dev. Biol., 271 (2004), p.~372.
	
	\bibitem{KIV14}
	{\sc M.~Kivel{\"a}, A.~Arenas, M.~Barth{\'e}lemy, J.~P. Gleeson, Y.~Moreno, and
		M.~A. Porter}, {\em Multilayer networks}, J. Complex Netw., 2 (2014),
	pp.~203--271, \url{https://doi.org/10.1093/comnet/cnu016},
	\url{https://arxiv.org/abs/http://comnet.oxfordjournals.org/content/2/3/203.full.pdf+html}.
	
	\bibitem{KOR18}
	{\sc B.~Korte and J.~Vygen}, {\em Combinatorial Optimization}, Springer,
	Berlin, Heidelberg, 2018, \url{https://doi.org/10.1007/978-3-642-24488-9}.
	
	\bibitem{LAD13}
	{\sc J.~Ladenbauer, J.~Lehnert, H.~Rankoohi, T.~Dahms, E.~Sch{\"o}ll, and
		K.~Obermayer}, {\em Adaptation controls synchrony and cluster states of
		coupled threshold-model neurons}, Phys. Rev. E, 88 (2013), p.~042713,
	\url{https://doi.org/10.1103/physreve.88.042713}.
	
	\bibitem{LEY18}
	{\sc I.~Leyva, I.~Sendi{\~n}a-Nadal, R.~Sevilla-Escoboza, V.~P. Vera-Avila,
		P.~Chholak, and S.~Boccaletti}, {\em Relay synchronization in multiplex
		networks}, Sci. Rep., 8 (2018), p.~8629.
	
	\bibitem{LEY17a}
	{\sc I.~Leyva, R.~Sevilla-Escoboza, I.~Sendi{\~n}a-Nadal, R.~Guti{\'e}rrez,
		J.~M. Buld{\'u}, and S.~Boccaletti}, {\em Inter-layer synchronization in
		non-identical multi-layer networks}, Sci. Rep., 7 (2017), p.~45475,
	\url{https://doi.org/10.1038/srep45475}.
	
	\bibitem{LIE15}
	{\sc J.~Liesen and V.~Mehrmann}, {\em Linear Algebra}, Springer, Cham, 2015,
	\url{https://doi.org/10.1007/978-3-319-24346-7}.
	
	\bibitem{MAK16}
	{\sc V.~A. Maksimenko, V.~V. Makarov, B.~K. Bera, D.~Ghosh, S.~K. Dana, M.~V.
		Goremyko, N.~S. Frolov, A.~A. Koronovskii, and A.~E. Hramov}, {\em Excitation
		and suppression of chimera states by multiplexing}, Phys. Rev. E, 94 (2016),
	p.~052205, \url{https://doi.org/10.1103/physreve.94.052205}.
	
	\bibitem{MIK18}
	{\sc M.~Mikhaylenko, L.~Ramlow, S.~Jalan, and A.~Zakharova}, {\em Weak
		multiplexing in neural networks: {S}witching between chimera and solitary
		states}, {C}haos, 29 (2019), p.~023122,
	\url{https://doi.org/10.1063/1.5057418}.
	
	\bibitem{MOT13a}
	{\sc A.~E. Motter, S.~A. Myers, M.~Anghel, and T.~Nishikawa}, {\em Spontaneous
		synchrony in power-grid networks}, Nat. Phys., 9 (2013), pp.~191--197,
	\url{https://doi.org/doi:10.1038/nphys2535}.
	
	\bibitem{MUC10}
	{\sc P.~J. Mucha, T.~Richardson, K.~Macon, M.~A. Porter, and J.~P. Onnela},
	{\em Community structure in time-dependent, multiscale, and multiplex
		networks}, Science, 328 (2010), pp.~876--878.
	
	\bibitem{MUL20}
	{\sc R.~Mulas, C.~Kuehn, and J.~Jost}, {\em Coupled dynamics on hypergraphs:
		Master stability of steady states and synchronization}, Phys. Rev. E, 101
	(2020), p.~062313, \url{https://doi.org/10.1103/physreve.101.062313}.
	
	\bibitem{NEW03}
	{\sc M.~E.~J. Newman}, {\em The structure and function of complex networks},
	SIAM Review, 45 (2003), pp.~167--256,
	\url{https://doi.org/10.1137/s0036144503}.
	
	\bibitem{NIC18}
	{\sc R.~Nicks, L.~Chambon, and S.~Coombes}, {\em Clusters in nonsmooth
		oscillator networks}, Phys. Rev. E, 97 (2018), p.~032213,
	\url{https://doi.org/10.1103/physreve.97.032213}.
	
	\bibitem{NIC17}
	{\sc V.~Nicosia, P.~S. Skardal, A.~Arenas, and V.~Latora}, {\em Collective
		phenomena emerging from the interactions between dynamical processes in
		multiplex networks}, Phys. Rev. Lett., 118 (2017), p.~138302.
	
	\bibitem{NIK19}
	{\sc D.~Nikitin, I.~Omelchenko, A.~Zakharova, M.~Avetyan, A.~L. Fradkov, and
		E.~Sch{\"o}ll}, {\em {C}omplex partial synchronization patterns in networks
		of delay-coupled neurons}, Phil. Trans. R. Soc. A, 377 (2019), p.~20180128.
	
	\bibitem{NIS06a}
	{\sc T.~Nishikawa and A.~E. Motter}, {\em Maximum performance at minimum cost
		in network synchronization}, Physica D, 224 (2006), pp.~77--89.
	
	\bibitem{OME19}
	{\sc I.~Omelchenko, T.~H{\"u}lser, A.~Zakharova, and E.~Sch{\"o}ll}, {\em
		Control of chimera states in multilayer networks}, Front. Appl. Math. Stat.,
	4 (2019), p.~67, \url{https://doi.org/10.3389/fams.2018.00067}.
	
	\bibitem{OME13}
	{\sc I.~Omelchenko, O.~E. Omel'chenko, P.~H{\"o}vel, and E.~Sch{\"o}ll}, {\em
		When nonlocal coupling between oscillators becomes stronger: patched
		synchrony or multichimera states}, Phys. Rev. Lett., 110 (2013), p.~224101,
	\url{https://doi.org/10.1103/physrevlett.110.224101}.
	
	\bibitem{OME10a}
	{\sc O.~E. Omel'chenko, M.~Wolfrum, and Y.~Maistrenko}, {\em Chimera states as
		chaotic spatiotemporal patterns}, Phys. Rev. E, 81 (2010), p.~065201(R),
	\url{https://doi.org/10.1103/physreve.81.065201}.
	
	\bibitem{PEC98}
	{\sc L.~M. Pecora and T.~L. Carroll}, {\em {M}aster {S}tability {F}unctions for
		{S}ynchronized {C}oupled {S}ystems}, Phys. Rev. Lett., 80 (1998),
	pp.~2109--2112, \url{https://doi.org/10.1103/physrevlett.80.2109}.
	
	\bibitem{PIK01}
	{\sc A.~Pikovsky, M.~Rosenblum, and J.~Kurths}, {\em Synchronization: a
		universal concept in nonlinear sciences}, Cambridge University Press,
	Cambridge, 1st~ed., 2001.
	
	\bibitem{PIT18}
	{\sc E.~Pitsik, V.~Makarov, D.~Kirsanov, N.~S. Frolov, M.~Goremyko, X.~Li,
		Z.~Wang, A.~E. Hramov, and S.~Boccaletti}, {\em Inter-layer competition in
		adaptive multiplex network}, New J. Phys., 20 (2018), p.~075004,
	\url{https://doi.org/10.1088/1367-2630/aad00d}.
	
	\bibitem{POI19}
	{\sc C.~Poignard, J.~P. Pade, and T.~Pereira}, {\em The effects of structural
		perturbations on the synchronizability of diffusive networks}, J. Nonlinear
	Sci., 29 (2019), pp.~1919--1942,
	\url{https://doi.org/10.1007/s00332-019-09534-7}.
	
	\bibitem{POR20}
	{\sc M.~A. Porter}, {\em Nonlinearity Networks: A 2020 Vision}, Springer
	International Publishing, Cham, 2020, ch.~6, pp.~131--159,
	\url{https://doi.org/10.1007/978-3-030-44992-6_6}.
	
	\bibitem{RAD13b}
	{\sc F.~Radicchi and A.~Arenas}, {\em Abrupt transition in the structural
		formation of interconnected networks}, Nat. Phys., 9 (2013), pp.~717--720,
	\url{https://doi.org/10.1038/nphys2761}.
	
	\bibitem{RAK20}
	{\sc S.~Rakshit, B.~K. Bera, E.~M. Bollt, and D.~Ghosh}, {\em Intralayer
		synchronization in evolving multiplex hypernetworks: Analytical approach},
	SIAM J. Appl. Dyn. Syst., 19 (2020), pp.~918--963,
	\url{https://doi.org/10.1137/18m1224441}.
	
	\bibitem{RAM19}
	{\sc L.~Ramlow, J.~Sawicki, A.~Zakharova, J.~Hlinka, J.~C. Claussen, and
		E.~Sch{\"o}ll}, {\em Partial synchronization in empirical brain networks as a
		model for unihemispheric sleep}, EPL, 126 (2019), p.~50007.
	
	\bibitem{REQ16}
	{\sc R.~J. Requejo and A.~D{\'i}az-Guilera}, {\em Replicator dynamics with
		diffusion on multiplex networks}, Phys. Rev. E, 94 (2016), p.~022301,
	\url{https://doi.org/10.1103/physreve.94.022301}.
	
	\bibitem{RES06b}
	{\sc J.~G. Restrepo, E.~Ott, and B.~R. Hunt}, {\em Characterizing dynamical
		importance of nodes and links}, Phys. Rev. Lett., 97 (2006),
	\url{https://doi.org/10.1103/physrevlett.97.094102}.
	
	\bibitem{ROE19a}
	{\sc V.~R{\"o}hr, R.~Berner, E.~L. Lameu, O.~V. Popovych, and S.~Yanchuk}, {\em
		Frequency cluster formation and slow oscillations in neural populations with
		plasticity}, PLoS ONE, 14 (2019), p.~e0225094,
	\url{https://doi.org/10.1371/journal.pone.0225094}.
	
	\bibitem{RUZ20}
	{\sc G.~Ruzzene, I.~Omelchenko, J.~Sawicki, A.~Zakharova, E.~Sch{\"o}ll, and
		R.~G. Andrzejak}, {\em Remote pacemaker-control of chimera states in
		multilayer networks of neurons}, Phys. Rev. E, 102 (2020), p.~052216,
	\url{https://doi.org/https://doi.org/10.1103/physreve.102.052216}.
	
	\bibitem{RYB19}
	{\sc E.~Rybalova, T.~Vadivasova, G.~Strelkova, V.~Anishchenko, and
		A.~Zakharova}, {\em Forced synchronization of a multilayer heterogeneous
		network of chaotic maps in the chimera state mode}, Chaos, 29 (2019),
	p.~033134.
	
	\bibitem{SAW18c}
	{\sc J.~Sawicki, I.~Omelchenko, A.~Zakharova, and E.~Sch{\"o}ll}, {\em {D}elay
		controls chimera relay synchronization in multiplex networks}, Phys. Rev. E,
	98 (2018), p.~062224.
	
	\bibitem{SAW18}
	{\sc J.~Sawicki, I.~Omelchenko, A.~Zakharova, and E.~Sch{\"o}ll}, {\em
		Synchronization scenarios of chimeras in multiplex networks}, Eur. Phys. J.
	Spec. Top., 227 (2018), p.~1161.
	
	\bibitem{SEM18}
	{\sc N.~Semenova and A.~Zakharova}, {\em Weak multiplexing induces coherence
		resonance}, Chaos, 28 (2018), p.~051104,
	\url{https://doi.org/10.1063/1.5037584}.
	
	\bibitem{SHA20b}
	{\sc M.~Shafiei, S.~Jafari, F.~Parastesh, M.~Ozer, T.~Kapitaniak, and M.~Perc},
	{\em Time delayed chemical synapses and synchronization in multilayer
		neuronal networks with ephaptic inter-layer coupling}, Commun. Nonlinear Sci.
	and Numer. Simul., 84 (2020), p.~105175,
	\url{https://doi.org/10.1016/j.cnsns.2020.105175}.
	
	\bibitem{SHA19}
	{\sc Y.~Shaverdi, S.~Panahi, T.~Kapitaniak, and S.~Jafari}, {\em Effect of
		intra-layer connection on the synchronization of a multi-layer cell network},
	Eur. Phys. J. ST, 228 (2019), pp.~2405--2417,
	\url{https://doi.org/10.1140/epjst/e2019-900051-9}.
	
	\bibitem{SOL13}
	{\sc L.~Sol{\'a}, M.~Romance, R.~Criado, J.~Flores, A.~Garcia~del Amo, and
		S.~Boccaletti}, {\em Eigenvector centrality of nodes in multiplex networks},
	Chaos, 23 (2013), p.~033131, \url{https://doi.org/10.1063/1.4818544}.
	
	\bibitem{SOL13a}
	{\sc A.~Sol{\'e}-Ribalta, M.~De~Domenico, N.~E. Kouvaris, A.~D{\'i}az-Guilera,
		S.~G\'omez, and A.~Arenas}, {\em Spectral properties of the laplacian of
		multiplex networks}, Phys. Rev. E, 88 (2013), p.~032807,
	\url{https://doi.org/10.1103/physreve.88.032807}.
	
	\bibitem{SOL16}
	{\sc A.~Sol{\'e}-Ribalta, S.~G\'omez, and A.~Arenas}, {\em Congestion induced
		by the structure of multiplex networks}, Phys. Rev. Lett., 116 (2016),
	p.~108701, \url{https://doi.org/10.1103/physrevlett.116.108701}.
	
	\bibitem{SOR18}
	{\sc D.~Soriano-Pa\~nos, L.~Lotero, A.~Arenas, and J.~G\'omez-Garde\~nes}, {\em
		Spreading processes in multiplex metapopulations containing different
		mobility networks}, Phys. Rev. X, 8 (2018), p.~031039,
	\url{https://doi.org/10.1103/physrevx.8.031039}.
	
	\bibitem{STE08}
	{\sc R.~A. Stefanescu and V.~K. Jirsa}, {\em A low dimensional description of
		globally coupled heterogeneous neural networks of excitatory and inhibitory
		neurons}, PLoS Comput Biol, 4 (2008), p.~e1000219,
	\url{https://doi.org/10.1371/journal.pcbi.1000219}.
	
	\bibitem{SUN09a}
	{\sc J.~Sun, E.~M. Bollt, and T.~Nishikawa}, {\em Master stability functions
		for coupled nearly identical dynamical systems}, Europhys. Lett., 85 (2009),
	p.~60011.
	
	\bibitem{TAN18}
	{\sc E.~Tang and D.~S. Bassett}, {\em Colloquium: Control of dynamics in brain
		networks}, Rev. Mod. Phys., 90 (2018), p.~031003,
	\url{https://doi.org/10.1103/revmodphys.90.031003}.
	
	\bibitem{TAN19}
	{\sc L.~Tang, X.~Wu, J.~L{\"u}, J.~Lu, and R.~M. D'Souza}, {\em Master
		stability functions for complete, intralayer, and interlayer synchronization
		in multiplex networks of coupled r{\"o}ssler oscillators}, Phys. Rev. E, 99
	(2019).
	\newblock 012304.
	
	\bibitem{VAI18}
	{\sc M.~Vaiana and S.~F. Muldoon}, {\em Multilayer brain networks}, J.
	Nonlinear Sci.,  (2018), pp.~1--23,
	\url{https://doi.org/10.1007/s00332-017-9436-8}.
	
	\bibitem{VAL15a}
	{\sc E.~Valdano, L.~Ferreri, C.~Poletto, and V.~Colizza}, {\em Analytical
		computation of the epidemic threshold on temporal networks}, Phys. Rev. X, 5
	(2015), p.~021005, \url{https://doi.org/10.1103/physrevx.5.021005}.
	
	\bibitem{VIR16c}
	{\sc Y.~S. Virkar, W.~L. Shew, J.~G. Restrepo, and E.~Ott}, {\em {F}eedback
		control stabilization of critical dynamics via resource transport on
		multilayer networks: {H}ow glia enable learning dynamics in the brain}, Phys.
	Rev. E, 94 (2016), p.~042310.
	
	\bibitem{WIN19}
	{\sc M.~Winkler, J.~Sawicki, I.~Omelchenko, A.~Zakharova, V.~Anishchenko, and
		E.~Sch{\"o}ll}, {\em Relay synchronization in multiplex networks of discrete
		maps}, EPL, 126 (2019), p.~50004.
	
	\bibitem{YAM20}
	{\sc M.~E. Yamakou, P.~G. Hjorth, and E.~A. Martens}, {\em Optimal self-induced
		stochastic resonance in multiplex neural networks: Electrical vs. chemical
		synapses}, Frontiers in Computational Neuroscience, 14 (2020), p.~62,
	\url{https://doi.org/10.3389/fncom.2020.00062}.
	
	\bibitem{YAN03b}
	{\sc S.~Yanchuk, Y.~Maistrenko, and E.~Mosekilde}, {\em Synchronization of
		time-continuous chaotic oscillators}, Chaos, 13 (2003), pp.~388--400.
	
	\bibitem{YAN11}
	{\sc S.~Yanchuk, P.~Perlikowski, O.~V. Popovych, and P.~A. Tass}, {\em
		Variability of spatio-temporal patterns in non-homogeneous rings of spiking
		neurons}, Chaos, 21 (2011), 047511, p.~047511,
	\url{https://doi.org/10.1063/1.3665200}.
	
	\bibitem{ZHA15a}
	{\sc X.~Zhang, S.~Boccaletti, S.~Guan, and Z.~Liu}, {\em Explosive
		synchronization in adaptive and multilayer networks}, Phys. Rev. Lett., 114
	(2015), p.~038701, \url{https://doi.org/10.1103/physrevlett.114.038701}.
	
\end{thebibliography}

\end{document}